\documentclass[10pt]{article}

\usepackage{amssymb, amsfonts, amsmath, amsthm, amsxtra}
\usepackage{thmtools, thm-restate}  % to allow repeating a theorem
\usepackage{graphics, graphicx}
\usepackage{old-arrows}
\usepackage{tabularx}
\usepackage{enumerate}
\usepackage{authblk}  % authors formatting
\usepackage{hyperref}
\usepackage{xcolor}  % define colors
\usepackage{titlesec}  % title style change
\usepackage[margin=1.1in]{geometry}
\usepackage{tcolorbox}

\definecolor{exodus}{RGB}{104, 109, 224}  % blue/purple
\definecolor{asbestos}{RGB}{127, 140, 141}  % light grey
\definecolor{gold}{RGB}{225, 177, 44}  % british gold
\definecolor{vanadyl}{RGB}{0, 151, 230}  % vanadyl blue

\hypersetup{
	colorlinks=true,
	citecolor=gold,
	linkcolor=vanadyl
}

\definecolor{clouds}{RGB}{236, 240, 241}
\definecolor{shy}{RGB}{162, 155, 254}
\definecolor{paradise}{RGB}{224, 253, 224}
\definecolor{electron}{RGB}{9, 132, 227}
\definecolor{alizarin}{RGB}{231, 76, 60}

\newcommand{\cc}[1]{\mathcal{#1}}  % calligraphic
\newcommand{\ctt}[1]{\texttt{#1}}  % text tt

\newcommand{\mtt}[2]{\texorpdfstring{#1}{#2}}  % maths in title (#1 latex, #2 pdf)

\newcommand{\NP}{\mathbf{NP}}  % NP (for complexity)
\newcommand{\st}{:}  % such that (could be \st)
\newcommand{\U}{X}  % universe
\newcommand{\pow}[1]{\mathbf{2}^{#1}}  % powerset
\newcommand{\card}[1]{\vert #1 \vert}  % cardinal of a set
\newcommand{\ftr}{\uparrow}  % filter (dualization)
\newcommand{\idl}{\downarrow}  % ideal (dualization)
\newcommand{\upp}{\,\scalebox{0.9}{$\nearrow$}\,}  % up-perspective
\newcommand{\dpp}{\,\scalebox{0.9}{$\searrow$}\,}  % down-perspective
\newcommand{\cl}{\phi}  % closure operator
\newcommand{\cs}{\cc{F}} % closure system
\newcommand{\Bp}{\cc{B}^+}  % dualization B^+
\newcommand{\Bm}{\cc{B}^-}  % dualization B^-
\newcommand{\imp}{\rightarrow}  % implication arrow
\newcommand{\is}{\Sigma}  % implicational base
\newcommand{\iscd}{\is_{\delta}}  % canonical direct base
\newcommand{\cd}{\delta}  % canonical direct relation
\newcommand{\G}{G}  % digraph of relations
\DeclareMathOperator{\Mi}{Mi}  % meet-irreducible elements
\DeclareMathOperator{\gen}{gen}  % generators
\DeclareMathOperator{\genD}{gen_{\mathit{D}}}  % D-generators
\declaretheorem[style=plain, name=Theorem]{thm} 
\declaretheorem[style=plain, name=Question]{ques}
\declaretheorem[style=plain, name=Corollary]{cor}
\declaretheorem[style=plain, name=Proposition]{prop}
\declaretheorem[style=plain, name=Lemma]{lem}

\declaretheorem[style=definition, name=Example]{exam}
\declaretheorem[style=definition, name=Remark]{rem}

\declaretheorem[style=remark, numbered=no, name=Proof sketch]{proofsketch}

\makeatletter
\newcommand{\problemtitle}[1]{\gdef\@problemtitle{#1}}% Store problem title
\newcommand{\probleminput}[1]{\gdef\@probleminput{#1}}% Store problem input
\newcommand{\problemquestion}[1]{\gdef\@problemquestion{#1}}% Store problem question
\NewEnviron{decproblem}{
	\problemtitle{}\probleminput{}\problemquestion{}% Default input is empty
	\BODY% Parse input
	\par\addvspace{.5\baselineskip}
	\noindent
	\begin{tabularx}{\textwidth}{@{\hspace{\parindent}} l X c}
		\multicolumn{2}{@{\hspace{\parindent}}l}{\@problemtitle} \\% Title
		\textbf{Input:} & \@probleminput \\% Input
		\textbf{Question:} & \@problemquestion% Output
	\end{tabularx}
	\par\addvspace{.5\baselineskip}}
\NewEnviron{genproblem}{
	\problemtitle{}\probleminput{}\problemquestion{}% Default input is empty
	\BODY% Parse input
	\par\addvspace{.5\baselineskip}
	\noindent
	\begin{tabularx}{\textwidth}{@{\hspace{\parindent}} l X c}
		\multicolumn{2}{@{\hspace{\parindent}}l}{\@problemtitle} \\% Title
		\textbf{Input:} & \@probleminput \\% Input
		\textbf{Output:} & \@problemquestion% Output
	\end{tabularx}
	\par\addvspace{.5\baselineskip}}
\makeatother 

\newcommand{\userparagraph}[1]{\paragraph{#1}}

\title{Computing the $D$-base and $D$-relation in finite closure systems}
\author[1]{Kira Adaricheva}
\author[2]{Lhouari Nourine}
\author[3]{Simon Vilmin}

\affil[1]{Department of Mathematics, Hofstra University, Hempstead, NY 11549, USA.}
\affil[2]{Université Clermont Auvergne, Clermont Auvergne INP, CNRS, LIMOS, F-63000 Clermont–Ferrand, France.}
\affil[3]{Aix-Marseille Université, CNRS, LIS, Marseille, France.}

\begin{document}
\maketitle

\begin{abstract}
Implicational bases (IBs) are a common representation of finite closure systems and lattices, along with meet-irreducible elements.
They appear in a wide variety of fields ranging from logic and databases to Knowledge Space Theory.

Different IBs can represent the same closure system.
Therefore, several IBs have been studied, such as the canonical and canonical direct bases.
In this paper, we investigate the $D$-base, a refinement of the canonical direct base.
It is connected with the $D$-relation, an essential tool in the study of free lattices.
The $D$-base demonstrates desirable algorithmic properties, and together with the $D$-relation, it conveys essential properties of the underlying closure system.
Hence, computing the $D$-base and the $D$-relation of a closure system from another representation is crucial to enjoy its benefits. 
However, complexity results for this task are lacking. 

In this paper, we give algorithms and hardness results for the computation of the $D$-base and $D$-relation.
Specifically, we establish the $\NP$-completeness of finding the $D$-relation from an arbitrary IB; we give an output-quasi-polynomial time algorithm to compute the $D$-base from meet-irreducible elements; and we obtain a polynomial-delay algorithm computing the $D$-base from an arbitrary IB.
These results complete the picture regarding the complexity of identifying the $D$-base and $D$-relation of a closure system.
\end{abstract}

%\newpage

\section{Introduction}

A \emph{closure system} over a finite groundset $\U$ is a set system containing $\U$ and closed under taking intersections.
The sets in a closure system are \emph{closed sets}, and when ordered by inclusion they form a \emph{lattice}.
Lattices and closure systems are used in a number of fields of mathematics and computer science such as algebra \cite{gratzer2011lattice}, databases \cite{mannila1992design}, logic \cite{hammer1995quasi}, or Knowledge Space Theory \cite{doignon1999knowledge} to mention a few.

Often, a closure system is implicitly given by one of the following two representations: \emph{meet-irreducible elements} or \emph{implicational bases} (IBs).
The family of meet-irreducible elements of a closure system is the unique minimal subset of closed sets from which the whole system can be rebuilt using intersections.
%In this paper, we are most interested in implications and IBs.
An implication over groundset $\U$ is a statement $A \imp c$ where $A$ is a subset of $\U$ and $c$ an element of $\U$.
In $A \imp c$, $A$ is the \emph{premise} and $c$ the \emph{conclusion}.
The implication $A \imp c$ stands for ``if a set includes $A$, it must also contain the element $c$''.
An implicational base (IB) over $\U$ is a collection of implications over $\U$.
An IB encodes a unique closure system.
On the other hand, a closure system can be represented by several equivalent IBs.
IBs are known as pure Horn CNFs in logic \cite{boros2009subclass, hammer1995quasi}, covers of functional dependencies in databases \cite{mannila1992design}, association rules in data mining \cite{agrawal1996fast}, or entailments in Knowledge Space Theory \cite{doignon1999knowledge}.

This variety of applications brought a rich theory of implications still at the core of several works, as witnessed by surveys \cite{bertet2018lattices,wild2017joy} and recent contributions \cite{berczi2024hypergraph, berczi2024matroid, bichoupan2022complexity, bichoupan2023independence, nourine2023hierarchical}. 
Specifically, since two distinct IBs can represent the same closure system, several IBs have been studied \cite{adaricheva2013ordered, bertet2010multiple, guigues1986familles, wild1994theory}.
Within this galaxy, the \emph{canonical base} and the \emph{canonical direct base} are arguably the two most shining stars.
They are unique, and they reflect two complementary ways of understanding IBs:
\begin{enumerate}[(1)]
\item \emph{the canonical base} \cite{guigues1986familles}, also known as the \emph{Duquenne-Guigues base}, puts emphasis on the premises of implications.
Its implications are of the form $A \imp C$ (rather than $A \imp c$) where $A$ must be \emph{pseudo-closed} and $C$ is the \emph{closure} of $A$, i.e., the set of all elements derived from $A$.
The canonical base has a minimum number of implications, and can be reached from any other IB in polynomial time \cite{guigues1986familles, day1992lattice}.

\item \emph{The canonical direct base} is defined by its conclusions (see the survey \cite{bertet2010multiple}).
It describes, for each element $c$, the inclusion-wise minimal sets $A$ such that $A \imp c$ holds in the closure system.
Such minimal sets are the \emph{minimal generators} of $c$.
The canonical direct base then gathers all implications of the form $A \imp c$ (known in logic as the prime implicates in Horn CNFs).
It enjoys the property of being \emph{direct}, meaning that the closure of a set can be computed with a single pass over the implications.
Moreover, the canonical direct base naturally captures the $\cd$-relation of lattice theory \cite{monjardet1997dependence}, where $c \cd a$ means that $a$ belongs to some minimal generator of $c$.
This relation and its transitive closure have been used implicitly in the context of pure Horn CNFs and implication-graphs for minimization \cite{boros1998horn} or to recognize and optimize acyclic Horn functions and some of their generalizations \cite{boros2009subclass, hammer1995quasi}.
\end{enumerate}
In this paper, we study a subset of the canonical direct base: the $D$-base \cite{adaricheva2013ordered}.
It relies on particular minimal generators called \emph{$D$-generators}.
A minimal generator $A$ of $c$ is a \emph{$D$-generator} of $c$ if its closure with respect to binary implications---implications of the form $a \imp c$---is minimal as compared to those of other minimal generators of $c$.
The $D$-generators also appear under the name \emph{minimal pairs} in the study of semimodular closure operators \cite{faigle1981projective}.
The $D$-base then consists in all implications of the form $A \imp c$ where $A$ is a $D$-generator of $c$, along with all valid (and non-trivial) binary implications.
In general, the $D$-base is much smaller than the canonical direct base.
Yet, it still enjoys directness as long as its implications are suitably ordered.
This makes the $D$-base appealing for computational purposes. 
Several application projects were carried out recently that provide analysis of data by employing the $D$-base \cite{adaricheva2015measuring, nation2021algorithms,adaricheva2023seabreese}.

Besides, the $D$-base embeds the $D$-relation of a closure system: $cDa$ holds if $a$ belongs to a $D$-generator of $c$.
The $D$-relation has played a major role in lattice theory since the 1970s \cite{JonsNat75}.
It is crucial in the study of free lattices \cite{freese1995free} and for the doubling of convex sets \cite{day1992double}.
Also, the $D$-relation compactly convey structural information of the underlying closure system.
The most striking examples are lower bounded closure systems that are precisely characterized by an acyclic $D$-relation \cite{freese1995free}.
Equivalently, lower bounded closure systems are known to be representable as lattices of suborders of partial orders \cite{sivak1978representations}.
Acyclic closure systems (i.e., acyclic Horn functions) \cite{adaricheva2017optimum, hammer1995quasi, wild1994theory} and closure systems generated by sub-semilattices of semilattices \cite{adaricheva1991ssl} are well-known examples of lower bounded closure systems.

Despite their importance, the algorithmic solutions for several questions related to $D$-base and $D$-relation are still lacking.
The first question is about computing the $D$-relation from an IB.

\begin{ques} \label{ques:D-relation}
Is it possible to compute in polynomial time the $D$-relation of a closure system given by an IB?
\end{ques}

The other two questions are two sides of the same coin for they ask to compute the $D$-base of a closure system given by either of the two representations previously mentioned, namely an IB or meet-irreducible elements.
Algorithms for such tasks are \emph{enumeration algorithms} for they aim at outputting, without repetitions, all the implications of the $D$-base.
However, the $D$-base may contain a number of implications being exponential in the size of the input IB or meet-irreducible elements.
Hence, we need to express the complexity of our algorithms in terms of the combined size of their input and their output.
This is \emph{output sensitive complexity} \cite{johnson1988generating, strozecki2019enumeration}.
More precisely, an enumeration algorithm runs in \emph{output-polynomial time} (resp.~ \emph{output-quasi-polynomial time}) if its execution time is polynomially (resp.~ quasi-polynomially) bounded by its input and output sizes.
A more refined description of an enumeration algorithm's complexity can be conducted by analyzing the regularity with which solutions are output. 
Specifically, if the delay between two solutions output and after the last one is polynomial in the size of the input, the algorithm runs with \emph{polynomial-delay}.
Note that if an algorithm runs with polynomial-delay, it runs in output-polynomial time.
We can now formalize the questions regarding the computation of the $D$-base.

\begin{ques} \label{ques:D-base-Mi}
Can the $D$-base of a closure system given by its meet-irreducible elements be computed in output-polynomial time?
\end{ques}

\begin{ques} \label{ques:D-base-IB}
Can the $D$-base of a closure system given by an IB be computed in output-polynomial time? 
\end{ques}

\userparagraph{Contributions.}
In this paper, we study the aforementioned questions.
Without loss of generality, our results are stated for \emph{standard} closure systems (see Section~\ref{sec:preliminaries}).
To begin with, we answer Question~\ref{ques:D-relation} negatively under different restrictions related to the acyclicity of $\delta$ or $D$.
Our results are summarized by the next two theorems.

\begin{thm}[restate=DACG, label=thm:D-relation-ACG]
The problem of identifying the $D$-relation given an IB is $\NP$-complete even for acyclic (i.e., $\delta$-acyclic) closure systems and IBs with premises of size at most 2.
\end{thm}

\begin{thm}[restate=DLB, label=thm:D-relation-LB]
The problem of identifying the $D$-relation given an IB is $\NP$-complete even for standard lower-bounded (i.e., $D$-acyclic) closure systems where $D$-paths have length at most $2$.
\end{thm}

For Question~\ref{ques:D-base-Mi}, we prove that the problem is equivalent to \emph{dualization in distributive closure systems}.
This task is a generalization of the famous \emph{hypergraph dualization} problem (see e.g., \cite{eiter1995identifying, eiter2008computational}), and has been shown recently to admit an output-quasi-polynomial time algorithm \cite{elbassioni2022dualization}.
Our results can thus be summarized as follows:

\begin{thm}[restate=DBDual, label=thm:D-base-dual]
There exists an output-polynomial time finding the $D$-base of a standard closure system given by its meet-irreducible elements if and only if there exists an output-polynomial time solving dualization in distributive closure systems.
\end{thm}

\begin{thm}[restate=DBMi, label=thm:D-base-Mi]
There is an output-quasi-polynomial time algorithm that computes the $D$-base of a standard closure system given by its meet-irreducible elements.
\end{thm}

Finally, we tackle Question~\ref{ques:D-base-IB} using the well-known \emph{solution graph traversal} method (see e.g., \cite{elbassioni2015polynomial, johnson1988generating}) and results rooted in database theory \cite{ennaoui2025polynomial}.
Specifically, we prove:

\begin{thm}[restate=DBIB, label=thm:D-base-IB]
There is a polynomial-delay algorithm that computes the $D$-base of a standard closure system given by an arbitrary IB.
\end{thm}

\userparagraph{Related work.}
We begin with complexity results regarding the $D$-relation.
In \cite{adaricheva2013ordered}, the authors provide sub and supersets of the $D$-relation that can be obtained in polynomial time from any IB.
On the one hand, it is possible to refine any IB to a subset of the $D$-base, which gives a subset of the $D$-relation.
On the other hand, the transitive closure of the $D$-relation can be recovered from any IB.
Similar results also exist for the $\cd$-relation \cite{boros1998horn}.
The hardness of computing the $\cd$-relation can be obtained as a corollary of the $\NP$-completeness of the prime attribute problem in databases \cite{lucchesi1978candidate}.
However, since $D \subset \delta$ in general, these results do not straightforwardly apply to the $D$-relation, especially when the structure of the closure system is to be studied.
Let us mention that when one is given meet-irreducible elements instead of an IB, both the $D$-relation and the $\cd$-relation can be computed in polynomial time by means of lattice-theoretic characterizations \cite{monjardet1997dependence, freese1995free}.
Still, to our knowledge, none of these results formally settles the complexity of computing the $D$-relation from an IB, especially with regard to the structure of the underlying closure system as we do in this paper.

We move to the tasks of computing the $D$-base from an IB or from meet-irreducible elements.
In \cite{rodriguez2015implicational, rodriguez2017formation}, the authors use simplification logic to come up with two algorithms computing the $D$-base from an arbitrary IB.
However, the complexity of these algorithms is not analyzed.
As for meet-irreducible elements as input, \cite{adaricheva2017discovery} give an algorithm based on hypergraph dualization.
Yet, their algorithm will produce in general a proper superset of the $D$-base.
Note that hypergraph dualization is a central open problem in algorithmic enumeration \cite{eiter1995identifying, eiter2008computational}. 
To date, the best algorithm for this task is due to Fredman and Khachiyan \cite{fredman1996complexity} and runs in output-quasi-polynomial time.
We finally mention the algorithm of Freese et al. \cite[Algorithm 11.12, p. 232]{freese1995free}  which computes the $D$-base from the whole closure system (in fact, a meet-join table of the lattice) as input.
As there is usually an exponential gap between a closure and either of its representation, IB or meet-irreducible elements, we cannot afford using this algorithm in our context. 

We also mention algorithms for finding the related canonical direct base.
If the input representation is meet-irreducible elements, it is well-known that the problem is equivalent to hypergraph dualization (see e.g., \cite{mannila1992design, wild2017joy}).
When an IB is given, the algorithm of \cite{lucchesi1978candidate} (see also \cite{berczi2024hypergraph}) listing the minimal keys of a closure system can be adapted to list the minimal generators of a single element with polynomial-delay.
Besides, the algorithm of \cite{boros1990polynomial} computes the whole canonical direct base with a total time linearly dependent on the number of implications in the output.
These techniques cannot be used out of the box to find the $D$-base as there may be an exponential gap between the two IBs. 
Instead, we adapt the procedure of \cite{ennaoui2025polynomial} to list the so called ideal-minimal keys---called $D$-minimal keys in our paper---of a closure system.
Nevertheless, our approach naturally extends results on minimal generators, all the while taking care of the binary implications playing an important role in the $D$-base.

\userparagraph{Paper Organization.}
In Section~\ref{sec:preliminaries} we give necessary definitions regarding closure systems and their representations.
We also formally define the problems we study in the paper.
In Section~\ref{sec:D-NPC} we establish the complexity of computing the $D$-relation.
Section~\ref{sec:D-base-Mi} is dedicated to the computation of the $D$-base of a closure system given by its meet-irreducible elements.
We investigate the dual problem of finding the $D$-base from an IB in Section~\ref{sec:D-base-IB}.
We conclude the paper with some remarks for future research in Section~\ref{sec:conclusion}.

\section{Preliminaries} \label{sec:preliminaries}

All the objects considered in this paper are finite.
For notions not defined here, we redirect the reader to \cite{gratzer2011lattice}.
If $\U$ is a set, $\pow{\U}$ is its powerset.
%We assume the reader is familiar with standard definitions regarding partially ordered sets and (directed) graphs.

\userparagraph{Closure operators, closure systems.} 
Let $\U$ be a set.
A \emph{closure operator} over $\U$ is a map $\cl \colon \pow{\U} \to \pow{\U}$ such that, for every $A, B \subseteq \U$: 
\begin{enumerate}[(1)]
    \item $A \subseteq \cl(A)$;
    \item $A \subseteq B$ implies $\cl(A) \subseteq \cl(B)$; and
    \item $\cl(\cl(A)) = \cl(A)$.
\end{enumerate}
The pair $(\U, \cl)$ is a \emph{closure space}.
A subset $F$ of $\U$ is \emph{closed} if $\cl(F) = F$.
For brevity, we write $\cl(a)$ instead of $\cl(\{a\})$ for $a \in \U$.
The family of closed sets of $\cl$ is called $\cs$.
We have $\cs = \{F \subseteq \U \st F = \cl(F)\} = \{\cl(A) \st A \subseteq \U\}$. 
We say that $\cl$ is \emph{standard} if %$\emptyset$ is closed and 
for every $a \in \U$, $\cl(a) \setminus \{a\}$ is closed. 
Consequently, $\emptyset$ is closed and $\cl(a) \neq \cl(b)$ unless $a = b$.
If $\cl(a) = \{a\}$ for each $a \in \U$, $(\U, \cs)$ is \emph{atomistic}.
Let $(\U, \cl)$ be a closure space and let $K \subseteq \U$.
We say that $K$ is a \emph{minimal spanning set} of its closure if $\cl(K') \subset \cl(K)$ for every $K' \subset K$.
If moreover $\cl(K) = \U$, $K$ is a \emph{minimal key} of the closure space.
For $c \in \U$, a subset $A \subseteq \U$ is a \emph{minimal generator} of $c$ if $c \in \cl(A)$ but $c \notin \cl(A')$ for every $A' \subset A$.
Any minimal generator of $c$ distinct from $c$ is called non-trivial.
We define $\gen(c)$ to be the non-trivial minimal generators of $c$.
Remark that minimal generators must be minimal spanning sets of their closure.

We move to closure systems.
A set system $(\U, \cs)$ is a \emph{closure system} if $\U \in \cs$ and $F_1 \cap F_2 \in \cs$ whenever $F_1, F_2 \in \cs$.
There is a well-known correspondence between closure operators and closure systems.

More precisely, if $(\U, \cl)$ is a closure space, the set system $(\U, \cs)$ is a closure system.
Dually, a closure system $(\U, \cs)$ induces a closure operator $\cl \colon \pow{\U} \to \pow{\U}$ given by $\cl(A) = \bigcap \{F \in \cs \st A \subseteq F\}$.
This correspondence is one-to-one.
The poset $(\cs, \subseteq)$ of closed sets ordered by inclusion is a \emph{(closure) lattice}.

\begin{rem}
In this paper, we always assume without loss of generality that closure systems are standard.
If $(\U, \cs)$ is a closure system, $\cl$ is its closure operator.
\end{rem}

Let $(\U, \cs)$ be a closure system with closure operator $\cl$.
We introduce a closure system devised from $(\U, \cs)$ that will be useful all along the paper.
Let $\cl^b$ be the closure operator defined by $\cl^b(A) = \bigcup_{a \in A} \cl(a)$.
It induces a closure system $(\U, \cs^b)$.
We have $\cs \subseteq \cs^b$.
Now let $M \in \cs$ such that $M \neq \U$.
The closed set $M$ is a \emph{meet-irreducible element} of $\cs$ if for every $F_1, F_2 \in \cs$, $F_1 \cap F_2 = M$ implies $F_1 = M$ or $F_2 = M$.
We denote by $\Mi(\cs)$ the meet-irreducible elements of $\cs$.
Remark that, in fact, $\Mi(\cs)$ is the set of meet-irreducible elements in the lattice-theoretic sense of the closure lattice $(\cs, \subseteq)$.
For every closed set $F$, we have $F = \bigcap \{M \in \Mi(\cs) \st F \subseteq M \}$.
Given $a \in \U$ and $M \in \Mi(\cs)$, we write $a \upp M$ if $a \notin M$ but $a \in F$ for every $F \in \cs$ such that $M \subset F$.
Equivalently $a \upp M$ if $M \in \max_{\subseteq}(\{M \in \Mi(\cs) \st a \notin M\})$.
Dually, we write $M \dpp a$ if $a \notin M$ but $\cl(a) \setminus \{a\} \subseteq M$.

\begin{exam}\label{ex:def-lattice}
Let $\U = \{1, 2, 3, 4, 5, 6\}$ and consider the closure system $(\U, \cs)$ whose closure lattice $(\cs, \subseteq)$ is depicted in Figure~\ref{fig:def-lattice}.
For convenience we write a set as the concatenation of its elements, e.g., $1234$ stands for $\{1, 2, 3, 4\}$.
For instance, $\cl(25) = 2456$.
We have $\Mi(\cs) = \{356, 13, 15, 1356, 1456, 124, \allowbreak 2456, 12456\}$.

\begin{figure}[ht!]
    \centering
    \includegraphics[scale=0.9]{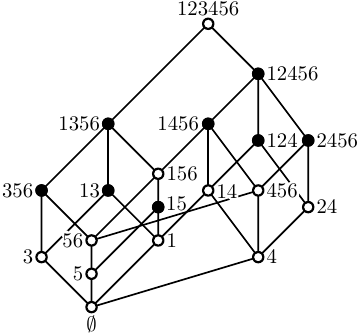}%
    \caption{The closure lattice of Example~\ref{ex:def-lattice}.
    Meet-irreducible elements are marked by black dots.}
    \label{fig:def-lattice}
\end{figure}

\end{exam}

We conclude the section with some classes of closure systems.
Let $(\U, \cs)$ be a closure system. 
If for every closed set $F \in \cs$, $F \neq \U$, there exists $a \notin F$ such that $F \cup \{a\} \in \cs$, then $(\U, \cs)$ is a \emph{convex geometry} \cite{edelman1985theory}.
Convex geometries are precisely the closure systems associated with \emph{anti-exchange} closure operators.
A closure operator $\cl$ is anti-exchange if for every closed set $F$ and $a, b \notin F$, $b \in \cl(F \cup \{a\})$ entails $a \notin \cl(F \cup \{b\})$.
Another useful characterizations of convex geometries relies on minimal spanning sets and extreme elements.
Let $A \subseteq \U$.
An element $a$ of $A$ is an \emph{extreme element} of $A$ if $a \notin \cl(A \setminus \{a\})$.
In particular if $A \in \cs$, $\cl(A \setminus \{a\}) = A \setminus \{a\}$ must hold.
Now a closure system is a convex geometry iff each closed set has a unique minimal spanning set which comprises precisely extreme elements.
For a given $A \subset \U$, we denote by $K_A$ the unique minimal spanning set of $\cl(A)$.
If $(\U, \cs)$ is a closure system and $F_1 \cup F_2 \in \cs$ for every $F_1, F_2 \in \cs$, then $(\U, \cs)$ (or its lattice $(\cs, \subseteq)$) is \emph{distributive}.
Every (standard) distributive closure system is a convex geometry.
In this case, we have $\Mi(\cs) = \{ \{c \in \U \st a \notin \cl(c)\} \st a \in \U\} $.
Hence, $\card{\Mi(\cs)} = \card{\U}$ and $\Mi(\cs)$ can be computed in polynomial time in the size of an IB $(\U, \is)$.
Observe that the closure system $(\U, \cs^b)$ defined previously is always distributive, regardless of the underlying closure system $(\U, \cs)$.
In the particular case when $(\U, \cs)$ is already distributive, we have $\cl = \cl^b$ and $\cs = \cs^b$.

\userparagraph{Implicational bases (IBs), $D$-generators, $D$-base.}  
We give some standard notations regarding implications and implicational bases (see also \cite{wild2017joy}).
An \emph{implication} over $\U$ is an expression $A \imp c$ where $A \cup \{c\} \subseteq \U$, also called \emph{unit implication} in \cite{bertet2010multiple}.
An implication is \emph{binary} if $A$ is a singleton.
For simplicity, we write them $a \imp c$ instead of $\{a\} \imp c$.
An \emph{implicational base} (IB) over $\U$ is a pair $(\U, \is)$ where $\is$ is a collection of implications over $\U$.
The number of implications of $\is$ is $\card{\is}$.
The \emph{total} size of $\is$ is $\sum \{\card{A} + \card{B} \st A \imp  B \in \is\}$.

An IB induces a closure operator $\cl$ where a subset $F$ of $\U$ is closed if for every implication $A \imp c \in \is$, $A \subseteq F$ implies $c \in F$.
The closure of $C \subseteq \U$ can be computed using the \emph{forward chaining procedure}.
This routine starts from $C$ and builds a sequence $C = C_0 \subseteq C_1 \subseteq \dots \subseteq C_m = \cl(C)$ of subsets of $\U$ such that, for every $1 \leq i \leq m$, $C_i = C_{i - 1} \cup \{c \in \U \st A \imp c \in \is, A \subseteq C_{i - 1}, c \notin C_{i-1}\}$.
A closure system admits several equivalent IBs.
An implication $A \imp c$ holds in a closure system if and only if $c \in \cl(A)$.

\begin{exam} \label{ex:def-IB}
An IB for the closure system of Example~\ref{ex:def-lattice} is $(\U, \is)$ with:
\[
\is = \{2 \imp 4, 6 \imp 5 \} \cup \{ 245 \imp 6, 34 \imp 1, 34 \imp 2, 34 \imp 5, 35 \imp 6, 45 \imp 6  \}
\]
\end{exam}

If an IB $(\U, \is)$ consists only of binary implications, the corresponding closure system is distributive.
Dually, if a closure system $(\U, \cs)$ is distributive, it can be represented by an IB with binary implications only.
Hence, its meet-irreducible elements and one of its IBs can be computed in polynomial time from one another.
For the distributive closure system $(\U, \cs^b)$ associated with a given $(\U, \cs)$, we can thus define the IB $(\U, \is^b)$ with $\is^b = \{a \imp c \st c \in \cl(a) \setminus \{a\}, a, c \in \U\}$ to represent $(\U, \cs^b)$.
It can be obtained in polynomial time from any IB $(\U, \is)$ of $(\U, \cs)$ or from $\Mi(\cs)$.

In this paper, we focus on the \emph{$D$-base} of a closure system $(\U, \cs)$ \cite{adaricheva2013ordered}.
It is an IB which, among equivalent IBs of a closure system, is uniquely defined.
In order to introduce the $D$-base, it is convenient to go through another unique IB: the \emph{canonical direct base} (see \cite{bertet2010multiple} for a survey).
It is the IB $(\U, \is_{\cd})$ devised from minimal generators of elements of $\U$, that is $\is_{\cd} = \{A \imp c \st A \in \gen(c), c \in \U\}$.

\begin{exam} \label{ex:def-cdb}
The IB of Example~\ref{ex:def-IB} is not the canonical direct base $(\U, \iscd)$ of the closure system of Example~\ref{ex:def-lattice}.
Indeed, $23$ is a minimal generator of $1$ but $23 \imp 1 \notin \is$.
Also, $245 \imp 6$ is in $\is$, but $245$ is not a minimal generator of $6$.
In fact, we have:
\[ 
\iscd = \left\{
\begin{array}{l l l l}
23 \imp 1, & 34 \imp 1, & 34 \imp 2, & 2  \imp 4, \\
23 \imp 5, & 34 \imp 5, &  6 \imp 5, & 23 \imp 6, \\
25 \imp 6, & 34 \imp 6, & 35 \imp 6, & 45 \imp 6
\end{array}
\right\}
\]
\end{exam}

To obtain the $D$-base from the canonical direct base, we need to use minimal generators and $\cl^b$.
A minimal generator $A$ of $c$ is a \emph{$D$-generator} of $c$ if $c \notin \cl^b(A)$ and for every minimal generator $A'$ of $c$, $A' \subseteq \cl^b(A)$ implies $A' = A$.
Observe that $c \notin \cl^b(A)$ implies $\card{A} \geq 2$ as $A$ is supposed to be a minimal generator of $c$.
%\ka{ Better: Observe that $c \notin \cl^b(A)$ implies $\card{A} \geq 2$.}
It will be convenient to call $\genD(c)$ the family of $D$-generators of $c$ (w.r.t. $(\U, \cs)$).
The $D$-base of the closure system $(\U, \cs)$ is $(\U, \is_D)$ where:

\begin{align*}
\is_D & = \is^b \cup \{A \imp c \st A \in \genD(c), c \in \U\} \\ 
 & = \{a \imp c \st c \in \cl(a) \setminus \{a\}, a, c \in \U\} \cup \{A \imp c \st A \in \genD(c), c \in \U\}
\end{align*}

\begin{rem}
Observe that a given $c \in \U$ does not need to admit $D$-generators, that is $\genD(c) = \emptyset$ is possible.
This happens if and only if $c$ has only singleton non-trivial minimal generators, i.e., if $c$ is \emph{join-prime} in $(\cs, \subseteq)$.
This is in turn equivalent to $\{a \in \U \st c \notin \cl(a)\}$ being the unique meet-irreducible element $M$ of $(\U, \cs)$ satisfying $c \upp M$.
Therefore, whether $\genD(c) = \emptyset$ can be tested in polynomial time from both $\Mi(\cs)$ and an IB $(\U, \is)$.
\end{rem}

\begin{exam}\label{ex:def-D}
Continuing Example~\ref{ex:def-cdb}, we can extract the $D$-base $(\U, \is_D)$ from $(\U, \iscd)$.
Note that $25$ is a minimal generator of $6$ which does not belong to $\genD(6)$ since $\cl^b(45) = 45 \subset 245 = \cl^b(25)$.
Similarly, every non-binary implication in $\iscd$ with element 2 in the premise is not in the $D$-base, because there is another implication where 2 is replaced by 4.
Finally:
\[
\is_D = \{2 \imp 4, 6 \imp 5\} \cup \{34 \imp 1, 34 \imp 2, 34 \imp 5, \\ 34 \imp 6, 35 \imp 6, 45 \imp 6 \}.  
\]
\end{exam}
In this paper we seek to solve the following two enumeration problems.

\begin{genproblem}
\problemtitle{$D$-base computation with meet-irreducible elements}
\probleminput{The meet-irreducible elements $\Mi(\cs)$ of a standard closure system $(\U, \cs)$}
\problemquestion{The $D$-base $(\U, \is_D)$ associated to $(\U, \cs)$}
\end{genproblem}

\begin{genproblem}
\problemtitle{$D$-base computation with implications}
\probleminput{An IB $(\U, \is)$ of a standard closure system $(\U, \cs)$}
\problemquestion{The $D$-base $(\U, \is_D)$ associated to $(\U, \cs)$}
\end{genproblem}

Observe that, by definition, the $D$-base is always a subset of the canonical direct base.
However, there may be an exponential gap between these two IBs in the general case, even though they coincide in atomistic closure systems.
\begin{exam}\label{ex:gap}
Let $\U = \{a_1, \dots, a_n, b_1, \dots, b_n, c\}$ and let $\is = \{a_i \imp b_i \st 1 \leq i \leq n\} \cup \{b_1 \dots b_n \imp c\}$.
The IB $(\U, \is)$ is the $D$-base of its closure system, while we have $\is_{\cd} = \{d_1 \dots d_n \imp c \st d_i \in \{a_i, b_i\}, 1 \leq i \leq n\} \cup \{a_i \imp b_i \st 1 \leq i \leq n\}$.
Thus, $\card{\is} = n + 1$ and $\card{\is_{\cd}} = 2^n + 1$.
\end{exam}
As a consequence, one cannot straightforwardly apply the algorithms for finding the canonical direct base to find the $D$-base.

\userparagraph{Relations $\cd$ and $D$.}
Minimal and $D$-generators induce two binary relations over $\U$.
More precisely, we put $c \cd a$ if $a$ belongs to a non-trivial minimal generator of $c$.
This is the $\cd$-relation.
Similarly, $cDa$ means that $a$ belongs to a $D$-generator of $c$.\footnote{The notations $\cd$ and $D$ are standard terminology in lattice theory \cite{adaricheva2013ordered,bertet2010multiple,freese1995free,monjardet1997dependence}}
The resulting relation is the $D$-relation.
By definition, $D \subseteq \cd$, and equality holds when $\is^b = \emptyset$, i.e., when the underlying closure system is atomistic.
However, in general, $D \subset \cd$ as illustrated in Figure~\ref{fig:def-relations}.
We note that both $\cd$ and $D$ admit characterizations by means of $\Mi(\cs)$ \cite{freese1995free,monjardet1997dependence}.
For $\delta$, $c \delta a$ holds if and only if there exists $M \in \Mi(\cs)$ such that $c \upp M$ and $a \notin M$ with $c \neq a$.
For $D$, $c D a$ holds if and only if there is $M \in \Mi(\cs)$ such that $c \upp M \dpp a$, again with $c \neq a$.
It follows that both $\delta$ and $D$ can be recovered in polynomial time from $\Mi(\cs)$.

There is a close connection between the relations $\cd$ and $D$ and \emph{implication-graphs} \cite{boros1990polynomial, hammer1995quasi}.
The \emph{implication-graph} of an IB $(\U, \is)$ is the directed graph $\G(\is) = (\U, E)$ where an arc $(a, c)$ belongs to $\G(\is)$ if there exists an implication $A \imp c$ in $\is$ such that $a \in A$.
Now, the $\cd$-relation precisely consists of the reversed arcs of $G(\is_{\cd})$.
Similarly, the $D$-relation is obtained from $\G(\is_D)$ by removing the arcs coming from binary implications and reversing the remaining ones.

\begin{exam}\label{ex:def-relations}
Figure~\ref{fig:def-relations} illustrates the $\cd$- and $D$-relations of the closure system of Example~\ref{ex:def-lattice}.
Note that $D$ is a proper subset of $\cd$.
\begin{figure}[ht!]
    \centering
    \includegraphics[page=2, scale=0.9]{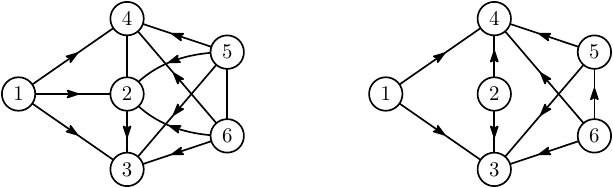}
    \caption{The $\cd$-relation (left) and the $D$-relation (right) of Example~\ref{ex:def-lattice} represented as directed graphs.
    An edge without arrow indicates arcs in both directions (such as between 2 and 4 on the left).}
    \label{fig:def-relations}
\end{figure}
\end{exam}
When the $\cd$-relation of a closure system $(\U,\is)$ does not have $\cd$-cycles, we say that $(\U, \is)$ is \emph{acyclic}.
For the closure system associated to $(\U, \is)$ to be acyclic, it is sufficient that $G(\is)$ is acyclic \cite{hammer1995quasi}.
Note that a standard distributive closure system is acyclic.
Acyclic closure systems are particular convex geometries that have been extensively studied as they correspond to acyclic Horn functions, poset type convex geometries and $G$-geometries \cite{boros2009subclass, hammer1995quasi, adaricheva2017optimum, wild1994theory}.
When the $D$-relation is acyclic the closure system is \emph{lower bounded}, or simply without $D$-cycles \cite{freese1995free, adaricheva2013ordered}.
This class strictly generalizes acyclic closure systems.
The closure system of Example~\ref{ex:def-lattice} is without $D$-cycles but it contains $\cd$-cycles.
Hence, it is lower bounded but not acyclic.

The third problem we study in this paper is the identification (or computation) of the $D$-relation from an IB.
Since the $D$-relation is a binary relation over elements of $\U$, this problem can be reduced to the problem of checking whether $cDa$ holds for each pair $a, c$ of $\U$:
\begin{decproblem}
\problemtitle{$D$-relation identification}
\probleminput{An IB $(\U, \is)$ of a standard closure system $(\U, \cs)$, and $a, c \in \U$.}
\problemquestion{Does $cDa$ hold?}
\end{decproblem}

\section{Finding the \mtt{$D$}{D}-relation from an implicational base} 
\label{sec:D-NPC}

In this section we establish the hardness of computing the $D$-relation of a closure system $(\U, \cs)$ given by an arbitrary IB $(\U, \is)$.
More precisely, we show that given $(\U, \is)$ and $a, c \in \U$, it is $\NP$-complete to check whether $cDa$ holds, even in acyclic closure systems or lower bounded closure systems with $D$-paths of constant size.
These are Theorems~\ref{thm:D-relation-ACG} and~\ref{thm:D-relation-LB} respectively.

As a preliminary step, we show that the problem belongs to $\NP$.
On this purpose, we rephrase in our terms a characterization of $D$-generators given in \cite[Algorithm 11.13, Equation 9, p.~230]{freese1995free}.
We could also use the expression of $D$ based on meet-irreducible elements as mentioned in the preliminaries, but the characterization of $D$-generators stated below will be useful in Section~\ref{sec:D-base-IB}.

\begin{lem}[{\cite[Equation 9, p.~230]{freese1995free}}] \label{lem:D-carac}
Let $(\U, \cs)$ be a closure system, $A \subseteq \U$ and $c \in \U \setminus A$.
Then, $A$ is a $D$-generator of $c$ if and only if $c \in \cl(A)$ and for every $a \in A$, $c \notin \cl(\cl^b(A) \setminus \{a\})$.
\end{lem}

\begin{cor} \label{cor:D-poly-cl}
Whether $A$ is indeed a $D$-generator of $c$ can be tested with at most $2\cdot \card{\U}$ calls to $\cl$.
\end{cor}

\begin{proof}
We apply Lemma~\ref{lem:D-carac}.
First, we compute $\cl^b(A) = \bigcup_{a \in A} \cl(a)$ with $\card{A}$ calls to $\cl$.
Then, for each $a \in A$, we check whether $c \in \cl(\cl^b(A) \setminus \{a\})$.
In total, we needed at most $2\cdot \card{\U}$ calls to $\cl$, which concludes the proof.
\end{proof}

The closure operator $\cl$ can be simulated in polynomial time given implications using forward chaining.
Note that it can also be computed in polynomial time from meet-irreducible elements using set intersections.
We deduce:

\begin{cor} \label{cor:DB-M}
The problem of deciding, given an IB $(\U, \is)$ and $a, c \in \U$, whether $cDa$ holds belongs to $\NP$.
\end{cor}

\begin{proof}
A certificate is a $D$-generator $A$ of $c$ such that $a \in A$.
Whether $A$ is indeed a $D$-generator of $c$ can be checked in polynomial time due to Corollary~\ref{cor:D-poly-cl}.
\end{proof}

We proceed to the hardness of identifying the $D$-relation in acyclic closure systems.
This is Theorem~\ref{thm:D-relation-ACG} which we restate below.
In fact, the hardness result also holds under the assumption that the premises of the input IB have constant size.
We give here the reduction, but the complete proof is delayed to \ref{app:D-relation-ACG} for clarity.

\DACG*

\begin{proofsketch}
We showed in Corollary~\ref{cor:DB-M} that the problem belongs to $\NP$.
To show $\NP$-hardness, we use a reduction from the problem $1$-in-$3$-SAT:
\begin{decproblem}
\problemtitle{1-in-3 SAT}
\probleminput{A positive $3$-CNF $\varphi = \{C_1, \dots, C_m\}$ over variables $V = \{v_1, \dots, v_n\}$}
\problemquestion{Is there a truth assignment of the variables in $V$ which is a model of $\varphi$ and such that each clause of $\varphi$ has exactly one true literal?}
\end{decproblem}
Let $\varphi = \{C_1, \dots, C_m\}$ be a non-trivial positive $3$-CNF over $V = \{v_1, \dots, v_n\}$.
Recall that a positive $3$-CNF is a conjunction of disjunctive clauses consisting of exactly three positive literals.
For simplicity, we view clauses as sets of variables.
Moreover, we identify an assignment of $V$ to the sets of its variables assigned 1.
Consequently, a 1-in-3 assignment of $V$ is a set of variables $T \subset V$ such that $\card{C_i \cap T} = 1$ for each clause $C_i$ of $\varphi$.
Two variables $v_i, v_j$ are in conflict if there exists a clause $C_k$ in $\varphi$ that contains both $v_i$ and $v_j$.
For each $1 \leq i \leq m$, we define $\varphi_i = \{C_i, \dots, C_m\}$.
A subset $T$ of $V$ is $\varphi$-conflict-free (resp.~ $\varphi_i$-conflict-free) if it does not contain variables in conflict among the clauses of $\varphi$ (resp.~ $\varphi_i$).

We build an instance of the $D$-relation identification problem.
Let $C = \{c_1, \dots, c_m, \allowbreak c_{m + 1}\}$ be a set of new gadget elements and let $\U = C \cup V$.
An element $c_i$, $1 \leq i \leq m$, represents the clause $C_i$.
The element $c_{m + 1}$ is another gadget element.
Now consider the set $\is = \is_{\mathit{mod}} \cup \is_{\mathit{conf}}$ of implications, where:
\begin{align*}
    \is_{\mathit{mod}} &  = \{c_i v_j \imp c_{i + 1} \st 1 \leq i \leq m \text{ and } v_j \in C_i\} \\
    \is_{\mathit{conf}} & = \{v_i v_j \imp c_{m + 1} \st v_i, v_j \in V \text{ and } v_i, v_j \text{ are in conflict in } \varphi\}
\end{align*}
Informally, $\is_{\mathit{mod}}$ models the fact that whenever the $i$-th clause is satisfied, we can proceed to the $(i+1)$-th clause or $c_{m + 1}$, and $\is_{\mathit{conf}}$ represents the conflicts induced by the variables appearing together in a clause.

We consider $(\U, \is)$ and we seek to identify whether $c_{m+1} D c_1$ holds.
Note that $(\U, \is)$ can be constructed in polynomial time in the size of $\varphi$ and $V$.
The premises of $\is$ have size $2$ and $G(\is)$ has no cycle.
Hence, $(\U, \cs)$ is an acyclic closure system represented by an IB with premises of size at most $2$ as required.
Moreover $\is$ has no binary implications so that minimal generators and $D$-generators coincide.
The proof in appendices relies on the characterization of minimal generators of $c_{m + 1}$. \qed
\end{proofsketch}

\begin{exam} \label{ex:ex-reduction-ACG}
We illustrate the reduction of Theorem~\ref{thm:D-relation-ACG}.
Let $V = \{1, \dots, 5\}$ and let $\varphi = \{C_1, C_2, C_3\}$ with $C_1 = \{2, 3, 4\}$, $C_2 = \{1, 2, 3\}$, $C_3 = \{1, 3, 5\}$.
According to the reduction of Theorem~\ref{thm:D-relation-ACG}, we have $\U = V \cup \{c_1, c_2, c_3, c_4\}$ and $\is = \is_{\mathit{mod}} \cup \is_{\mathit{conf}}$ where:
\[
\is_{\mathit{mod}} = \left\{
\begin{array}{l l l}
2 c_1 \imp c_2, & 3 c_1 \imp c_2, & 4 c_1 \imp c_2, \\
1 c_2 \imp c_3, & 2 c_2 \imp c_3, & 3 c_2 \imp c_3, \\
1 c_3 \imp c_4, & 3 c_3 \imp c_4, & 5 c_3 \imp c_4
\end{array}
\right\}
\text{ and } \;
\is_{\mathit{conf}} = \left\{
\begin{array}{l l}
12 \imp c_4, & 13 \imp c_4, \\ 
15 \imp c_4, & 23 \imp c_4, \\ 
24 \imp c_4, & 34 \imp c_4, \\ 
35 \imp c_4
\end{array}
\right\}
\]
We have $\is \subseteq \is_D$, so that the $D$-relation of the corresponding closure system is a subset (once reversed) of $G(\is)$.
This subset is sufficient to recover the transitive closure of $D$.
We represent it in Figure~\ref{fig:ex-reduction-ACG}.
\begin{figure}[ht!]
    \centering
    \includegraphics[scale=0.9, page=1]{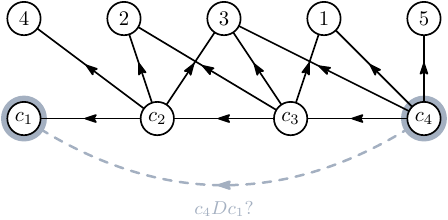}%
    \caption{The subset of the $D$-relation of $\is$ in Example~\ref{ex:ex-reduction-ACG}.
    This subset is sufficient to recover the transitive closure of $D$.
    Whether $c_4 D c_1$ holds (highlighted in grey) depends on whether the associated formula $\varphi$ has a $1$-in-$3$ assignment.}
    \label{fig:ex-reduction-ACG}
\end{figure}
Deciding whether $c_4Dc_1$ holds amounts to decide whether $\varphi$ has a 1-in-3 valid assignment.
Here, $14$ is such an assignment so that $14c_1$ is a $D$-generator of $c_4$, and $c_4 D c_1$ holds true.
\end{exam}

Remark that our reduction shows at the same time that testing $c \cd a$ is hard in acyclic closure systems, which was also unknown to our knowledge.
Understanding what makes hard to decide whether $c D a$ holds is an intriguing question.
The reduction of Theorem~\ref{thm:D-relation-ACG} and Example~\ref{ex:ex-reduction-ACG} suggest the difficulty comes from the fact that two elements can be separated by an arbitrary long $D$-path.
However, we prove in the next theorem that the problem remains $\NP$-complete even when the length of a $D$-path is bounded.
Interestingly though, this complexity comes at the cost of cycles in the $\cd$-relation and larger premises.
Again, we give the reduction, but leave the complete proof in \ref{app:D-relation-LB} for readability.

\DLB*

\begin{proofsketch}
To show $\NP$-hardness, we use another reduction from $1$-in-$3$ SAT.
Let $\varphi = \{C_1, \dots, C_m\}$ be a positive $3$-CNF over variables $V = \{v_1, \dots, v_n\}$.
Let $C = \{c_1, \dots, c_m\}$.
We use the same terminology as in Theorem~\ref{thm:D-relation-ACG}.
Let $\U = V \cup C \cup \{a, b\}$ be a new groundset, and let $\is = \is_{\mathit{mod}} \cup \is_{\mathit{conf}} \cup \is^b$ where:
\begin{align*}
\is_{\mathit{mod}} & = \{a v_i \imp c_j \st v_i \in C_j \} \cup \{C \imp b\} \\
\is_{\mathit{conf}} & = \{v_i v_j \imp b \st v_i, v_j \in V \text{ and } v_i, v_j \text{ are in conflict in } \varphi \} \\
\is^b & = \{c_i \imp a \st c_i \in C\}
\end{align*}
Informally, $\is_{\mathit{conf}}$ models all the conflicts in between two variables appearing in a clause and $\is_{\mathit{mod}}$ models the fact that taking one element in a clause together with $a$ for all the clauses yields a $1$-in-$3$ assignment of $\varphi$.
Finally, the binary implications $\is^b$ guarantee that $a$ will never appear in any minimal generator together with one of the $c_i$'s.
Finally, let $(\U, \cs)$ be the closure system associated to $(\U, \is)$.
We call $\cl$ the corresponding closure operator.
The reduction can be conducted in polynomial time.
The rest of the proof shows that: (1) $(\U, \is)$ and $(\U, \cs)$ meet the requirements of the statement, and (2) 
$bDa$ holds if and only if $\varphi$ has a valid 1-in-3 assignment. \qed
\end{proofsketch}

\begin{exam}\label{ex:ex-reduction-LB}
We illustrate the reduction on the same formula as in Example~\ref{ex:ex-reduction-ACG}.
Here, we will have $C = \{c_1, c_2, c_3\}$ and $\U = V \cup C \cup \{a, b\}$ and $\is = \is_{\mathit{mod}} \cup \is_{\mathit{conf}} \cup \is_a$ where $\is^b = \{c_1 \imp a, c_2 \imp a, c_3 \imp a\}$ and:
\[
\is_{\mathit{mod}} = \left\{
\begin{array}{l l l}
4 a \imp c_1, & 2 a \imp c_1, & 3 a \imp c_1, \\
2 a \imp c_2, & 3 a \imp c_2, & 1 a \imp c_2, \\
3 a \imp c_3, & 1 a \imp c_3, & 5 a \imp c_3, \\
c_1 c_2 c_3 \imp b
\end{array}
\right\}
\text{, } \;
\is_{\mathit{conf}} = \left\{
\begin{array}{l l}
12 \imp b, & 13 \imp b, \\ 
15 \imp b, & 23 \imp b, \\ 
24 \imp b, & 34 \imp b, \\ 
35 \imp b
\end{array}
\right\}.
\]
Again, $\is \subseteq \is_D$ so reversing the arcs of $G(\is)$ and dropping arcs of the binary implications in $\is^b$ gives a subset of the $D$-relation, see Figure~\ref{fig:ex-reduction-LB}.
\begin{figure}[ht!]
    \centering
    \includegraphics[scale=0.9, page=2]{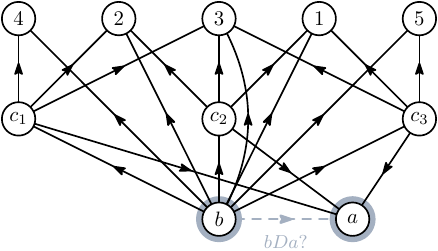}%
    \caption{The subset of the $D$-relation of $\is$ in Example~\ref{ex:ex-reduction-LB}.
    The only possible missing pair in the $D$-relation is $b D a$.
    Observe that the $D$-relation is acyclic and that all paths have size at most $2$.
    Whether $b D a$ holds (highlighted in grey) depends on whether the associated formula $\varphi$ has a $1$-in-$3$ assignment.}
    \label{fig:ex-reduction-LB}
\end{figure}

\end{exam}

Whether the $D$-relation can be computed in polynomial time when restricted to acyclic closure systems with $D$-paths of constant size remains an open question for future work.

\section{Finding the \mtt{$D$}{D}-base from meet-irreducible elements} \label{sec:D-base-Mi}

We prove that there exists an output-polynomial time algorithm for computing the $D$-base given meet-irreducible elements if and only if there is an output-polynomial time algorithm for dualization in distributive lattices, i.e., Theorem~\ref{thm:D-base-dual}.
Building on the algorithm of \cite{elbassioni2022dualization} that solves this latter problem in output-quasi-polynomial time, we deduce that the $D$-base of a closure system given by its meet-irreducible elements can be computed in output-quasi-polynomial time too.
This is Theorem~\ref{thm:D-base-Mi}.

We first give some definitions related to dualization.
Let $(\U, \cs)$ be a closure system.
An \emph{antichain} of the lattice $(\cs, \subseteq)$ is a subset $\cc{B}$ of $\cs$ consisting of pairwise incomparable closed sets.
We denote by $\idl \cc{B}$ the (order) \emph{ideal} of $\cc{B}$ in $(\cs, \subseteq)$, i.e., $\idl  \cc{B} = \{F \in \cs \st F \subseteq B \text{ for some } B \in \cc{B}\}$.
The (order) \emph{filter} of $\cc{B}$, denoted $\ftr \cc{B}$, is defined dually.
Let $\Bm$, $\Bp$ be two antichains of $(\cs, \subseteq)$.
We say that $\Bm$ and $\Bp$ are \emph{dual} in $(\cs, \subseteq)$ if $\ftr \Bm \cup \idl \Bp = \cs$ and $\ftr \Bm \cap \idl \Bp = \emptyset$.
For each antichain $\Bp$, there exists a unique dual antichain $\Bm$.
In particular, we have $\Bm = \min_{\subseteq}(\{F \in \cs \st F \nsubseteq B, \text{ for all } B \in \Bp\})$.
Symmetrically, each antichain $\Bm$ has a unique dual antichain $\Bp$.
The problem of dualizing distributive lattices now reads as follows. 

\begin{genproblem}
\problemtitle{Distributive lattice dualization}
\probleminput{An IB $(\U, \is)$ of a distributive closure system $(\U, \cs)$ and an antichain $\Bp$ of $(\cs, \subseteq)$.}
\problemquestion{The antichain $\Bm$ dual to $\Bp$ in $(\cs, \subseteq)$.}
\end{genproblem}

This problem is a generalization of hypergraph dualization (see \cite{eiter1995identifying, eiter2008computational} for further details on this task).
In hypergraph dualization, $\is$ is empty and the corresponding distributive lattice is in fact the powerset (Boolean) lattice $(\pow{\U}, \subseteq)$.
Since one can translate between the meet-irreducible elements and an IB of a distributive closure system in polynomial time, the input to dualization in \emph{distributive} lattices can be considered to be either of the two representations, or even both. 

\begin{exam} \label{ex:ex-dualization}
Let $\U = \{1, 2, 3, 4, 5\}$ and consider the (distributive) closure system $(\U, \cs)$ given by the IB $(\U, \is)$ with $\is$:
\[ 
\is = \{3 \imp 2, 2 \imp 1, 3 \imp 1, 5 \imp 4\}
\]
The corresponding closure lattice is depicted in Figure~\ref{fig:ex-dualization}.
Remark that $\is = \is^b$.
The antichains $\Bm = \{123, \allowbreak 124, 145\}$ and $\Bp = \{12, 14, 45\}$ are dual in $(\cs, \subseteq)$.
They are highlighted in Figure~\ref{fig:ex-dualization}.
\begin{figure}[ht!]
    \centering
    \includegraphics[scale=0.9]{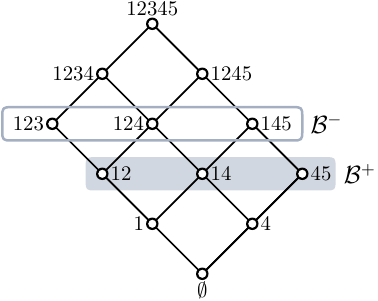}
    \caption{The distributive lattice of Example~\ref{ex:ex-dualization} with two dual antichains $\Bp$, $\Bm$.}
    \label{fig:ex-dualization}
\end{figure}

\end{exam}

In view of Theorem~\ref{thm:D-base-dual}, we first reduce dualization of distributive lattices to the problem of computing the $D$-base of a closure system given by its meet-irreducible elements.
Let $(\U, \is)$ be a binary IB of a distributive closure system $(\U, \cs)$.
%We assume without loss of generality that $\is = \is^b$.
Recall from Section~\ref{sec:preliminaries} that $\cl = \cl^b$ as $(\U, \cs)$ is distributive.
Moreover, recall that distributive closure systems are acyclic convex geometries.
Hence, each closed set $F \in \cs$ has a unique minimal spanning set $K_F$, consisting of the extreme elements of $F$.
Let $\Bp$ be an (non-empty) antichain of $(\cs, \subseteq)$, and let $\Bm$ be its dual antichain.

Let $\U' = \U \cup \{d\}$ where $d$ is a new gadget element.
We first construct a closure system $(\U', \cs')$ by means of its meet-irreducible elements which we define as follows:
\[ \Mi(\cs') = \Bp \cup \{M \cup \{d\} \st M \in \Mi(\cs)\} \]
Note that $\Mi(\cs')$ can be constructed in polynomial time from $(\U, \is)$ and $\Bp$, and that it is a well-defined collection of meet-irreducible elements.
Indeed, each $B \in \Bp$ is incomparable to every other set of $\Bp$ and other sets in $\Mi(\cs')$ always contain $d$.
So $B$ cannot be nontrivially obtained as the intersection of different sets in $\Mi(\cs')$.
The fact that $M \cup \{d\}$ is meet-irreducible comes from the observations that $d \in M \cup \{d\}$ and $M \in \Mi(\cs)$.
We call $\cl'$ the closure operator associated to $(\U', \cs')$.
We readily have that $\cs' = \{F \in \cs \st F \in \idl \Bp \text{ in } (\cs, \subseteq) \} \cup \{F \cup \{d\} \st F \in \cs\}$ and, for each $A \subseteq \U'$,
\[ 
\cl'(A) = \begin{cases}
\cl(A) & \text{if } d \notin A \text{ and } \cl(A) \in \idl \Bp \text{ in } (\cs, \subseteq) \\
\cl(A \setminus \{d\}) \cup \{d\} & \text{otherwise.}
\end{cases}
\]
Notice that $(\U', \cs')$ is standard as long as $(\U, \cs)$ is.
Moreover, for each $a \in \U'$, $\cl^b(a) \subseteq \cl'^b(a)$ holds.
Thus, $\cl(A) = \cl^b(A) \subseteq \cl'^b(A) \subseteq \cl'(A)$ for all $A \subseteq \U$.
If $a \in \U$ and $d \notin \cl'^b(a)$, then $\cl^b(a) = \cl'^b(a)$.
As a consequence, for every $A \subseteq \U'$ such that $d \notin \cl'(A)$, $\cl(A) = \cl^b(A) = \cl'^b(A)$.
This can be observed in Figure~\ref{fig:red-dual} of Example~\ref{ex:red-dual}.
We proceed to show that the antichain $\Bm$ dual to $\Bp$ in $(\cs, \subseteq)$ can be recovered in polynomial time from the $D$-base $(\U', \is'_D)$ of $(\U', \cs')$.
We begin with a useful claim regarding the nature of $(\U', \cs')$.

\begin{prop} \label{prop:Fp-CG}
The closure system $(\U', \cs')$ is a convex geometry.
Therefore, each closed set $F' \in \cs'$ has a unique minimal spanning set $K'_F$.
\end{prop}

\begin{proof}
To show that $(\U', \cs')$ is a convex geometry, we prove that for each closed set $F' \neq \U'$, there is $a \in \U' \setminus F'$ such that $F' \cup \{a\}$ is closed.
If $d \notin F'$, then $d = a$ meets the requirement.
Assume $d \in F'$.
Then, $F'$ lies in $\{F \cup \{d\} \st F \in \cs\}$.
Put $F' = F \cup \{d\}$ with $F \in \cs$.
Since $(\U, \cs)$ is distributive, it is a convex geometry and as $F' \neq \U'$, there exists $a \in \U \setminus F$ such that $F \cup \{a\} \in \cs$ so that $F' \cup \{a\} \in \cs'$ holds.
We deduce that $(\U', \cs')$ is a convex geometry, which concludes the proof.
\end{proof}

Then, we characterize the $D$-base of $(\U', \cs')$.
Recall that since $(\U, \cs)$ is distributive, $\is = \is^b$ faithfully represents $(\U, \cs)$.
Henceforth, every non-trivial minimal generator in $(\U, \cs)$ is a singleton, meaning that there is no $D$-generator in $(\U, \cs)$.
In the next statement we characterize the $D$-base of $(\U', \is'_D)$.
The proof is given in \ref{app:D-base-dual-DB} for readability.

\begin{prop}[restate=PDB, label=prop:D-base-dual-DB]
The $D$-base $(\U', \is'_D)$ of $(\U', \cs')$ is defined by:
\[ 
\is'_D =  \is'^b \cup \{K'_B \imp d \st B \in \Bm, \card{K'_B} \geq 2\}
\]
where $\is'^b = \is^b \cup \{a \imp d \st B \subseteq \cl(a) \text{ for some } B \in \Bm \}$.
\end{prop}

\begin{prop} \label{prop:D-base-dual-poly}
The $D$-base $(\U', \is'_D)$ has polynomial size with respect to the sizes of $(\U, \is)$ and $\Bm$.
\end{prop}

\begin{proof}
This follows from Proposition~\ref{prop:Fp-CG} that for all $B \in \Bm$, $\cl'(B)$ has a unique minimal spanning set.
\end{proof}

We are in position to conclude the first side of the reduction.

\begin{lem}\label{lem:D-base-harder-dual}
There is an output-polynomial time algorithm for dualization in distributive lattices if there is one for computing the $D$-base of a closure system given by its meet-irreducible elements.
\end{lem}

\begin{proof}
Let \ctt{A} be an output-polynomial time algorithm computing the $D$-base of a closure system given its meet-irreducible elements.
Given an instance $(\U, \is)$ and $\Bp$ of distributive lattice dualization, we show how to compute $\Bm$ in output-polynomial time using \ctt{A}.
First, we build $\U'$ and $\Mi(\cs')$ in polynomial time.
Then, we use \ctt{A} to compute $(\U', \is'_D)$ in time polynomial with respect to $(\U, \is), \Bp$ and $\Bm$ according to Proposition~\ref{prop:D-base-dual-poly}.
Based on Proposition~\ref{prop:D-base-dual-DB}, we can recover a subset of $\Bm$ by computing $\cl(K'_B)$ in polynomial time for each $K'_B \imp d$ in $\is'_D$.
Again by Proposition~\ref{prop:D-base-dual-DB}, the only missing elements of $\Bm$ are of the form $\cl(a)$ for some $a \in \U$ if there is any.
To identify such elements, we go through the following operations:
\begin{itemize}
    \item find the elements $a \in \U$ such that $d \in \cl'(a)$ but $K'_B \nsubseteq \cl'(a)$ for each $K'_B$ appearing in the non-binary part of $\is'_D$.
    This can be done in time polynomial in the size of $\Bm$ and $\Mi(\cs')$, hence in output-polynomial time.
    \item among these elements, find those whose closure with respect to $\cl'^b$ is inclusion-wise minimal.
    This can be done in polynomial time in the size of $\Mi(\cs')$ and $\U'$.
\end{itemize}
Computing the closure in $\cl$ of the resulting elements will produce the remaining parts of $\Bm$.
As long as \ctt{A} runs in output-polynomial time, the whole procedure correctly output the antichain $\Bm$ dual to $\Bp$ in $(\cs, \subseteq)$ in output-polynomial time.
This concludes the proof.
\end{proof}

We illustrate the reduction with the closure lattice of Example~\ref{ex:ex-dualization}.

\begin{exam} \label{ex:red-dual}
In Example~\ref{ex:ex-dualization} we have $\Bp = \{12, 14, 45\}$ and $\Bm = \{123, 124, 145\}$ (see Figure~\ref{fig:ex-dualization}).
As in the reduction we define $\U' = \U \cup \{d\}$ and the meet-irreducible elements of the closure system $(\U', \cs')$:
\[ 
\Mi(\cs') = \Bp \cup \{M \cup \{d\} \st M \in \Mi(\cs)\} = \{12, 14, 45, 45d, 123d, 145d, 1234d, 1245d\} 
\] 
Following Proposition~\ref{prop:D-base-dual-DB}, the $D$-base of $(\U', \cs')$ is given by:
\[
\is'_D = \{2 \imp 1, 3 \imp 2, 3 \imp 1, 3 \imp d, 5 \imp 4 \} \cup \{24 \imp d, 15 \imp d\} 
\]

We represent the closure lattice $(\cs', \subseteq)$ in Figure~\ref{fig:red-dual}.
Finding the $D$-base $(\U', \is'_D)$ from $\Mi(\cs')$ amounts to find (polynomially many) binary implications plus the unique minimal spanning sets of each element of $\Bm$ not covered by binary implications.
Here, the members $124$ and $145$ of $\Bm$, which corresponds to $124d$ and $145d$ in $(\cs', \subseteq)$ are generated by non-binary implications of the $D$-base, namely $24 \imp d$ and $15 \imp d$.
The last element of $\Bm$, $123$, corresponding to $123d$, is identified by $3 \imp d$.
\begin{figure}[ht!]
\centering 
\includegraphics[scale=0.9, page=2]{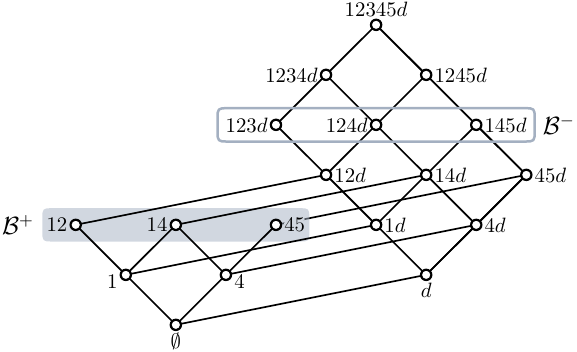}%
\caption{The lattice  $(\cs', \subseteq)$ resulting from the reduction of Lemma~\ref{lem:D-base-harder-dual} applied to Example~\ref{ex:ex-dualization}. 
Observe that $\Bp$ is now part of $\Mi(\cs')$.}
\label{fig:red-dual}
\end{figure}
\end{exam}

We turn our attention to the other side of Theorem~\ref{thm:D-base-dual}, namely, that an output-polynomial time algorithm for solving distributive lattice dualization yields an output-polynomial time algorithm for computing the $D$-base of a closure system given by its meet-irreducible elements.
To do so, we prove that finding the $D$-base of a closure system $(\U, \cs)$ given by $\Mi(\cs)$ reduces to (at most) $\card{\U}$ instances of distributive lattice dualization.
More precisely, we show that given $c \in \U$, we can recover $\genD(c)$ from the antichain dual to $\{M \in \Mi(\cs) \st c \upp M\}$ in the distributive lattice $(\cs^b, \subseteq)$.
We formalize this idea in Lemma~\ref{lem:gen-D-dual}, and we illustrate it via Example~\ref{ex:Dual-D-base}.

\begin{lem} \label{lem:gen-D-dual}
Let $(\U, \cs)$ be a closure system.
For $c \in \U$, the two following antichains are dual in $(\cs^b, \subseteq)$:
\begin{itemize}
    \item $\Bp = \{M \in \Mi(\cs) \st c \upp M\}$
    \item $\Bm = \{\cl^b(A) \st A \in \genD(c)\} \cup \{\cl^b(c)\}$
\end{itemize}
Moreover, $\cl^b$ is a one-to-one mapping between $\genD(c)$ and $\{\cl^b(A) \st A \in \genD(c)\}$.
\end{lem}

\begin{proof}
First, we prove that $\Bm$ and $\Bp$ are indeed antichains of $(\cs^b, \subseteq)$.
Since $\cs \subseteq \cs^b$, each $M \in \Mi(\cs)$ belongs to $\cs^b$.
Moreover, by definition of $\upp$, all the meet-irreducible elements of $\Bp$ must be incomparable.
Thus, $\Bp$ is an antichain of $(\cs^b, \subseteq)$.
On the other hand, $\Bm \subseteq \cs^b$ holds by definition.
The fact that $\Bm$ is an antichain follows from the definition of $D$-generators.

We proceed to show that $\Bm$ and $\Bp$ are dual.
First, we prove that $\ftr \Bm \cap \idl \Bp = \emptyset$.
It is sufficient to show that $\Bm \cap \idl \Bp = \emptyset$, i.e., that for each $B \in \Bm$, $B \nsubseteq M$ for every $M \in \Bp$.
Since $\Bp = \{M \in \Mi(\cs) \st c \upp M\}$ and $\cl(c) = \cl^b(c)$, we readily have that $\cl^b(c) \nsubseteq M$ for each $M \in \Bp$. %\ka{I replaced two $\Bm$ by two $\Bp$ in the previous sentence.}
We consider $\genD(c)$.
Note that if $\genD(c) = \emptyset$, then $\ftr \Bm \cap \idl \Bp = \emptyset$ is already proved.
Hence, assume $\genD(c) \neq \emptyset$ and let $A \in \genD(c)$.
We have $A \in \gen(c)$.
Thus, for every $M \in \Mi(\cs)$ such that $c \upp M$, there exists $a \in A \setminus M$.
Consequently, $\cl^b(A) \nsubseteq \cl^b(M) = M$.
We derive that $\ftr \Bm \cap \idl \Bp = \emptyset$ as required.

It remains to show that $\ftr \Bm \cup \idl \Bp = \cs^b$.
It is sufficient to show that for every $F \in \cs^b$ such that $F \notin \idl \Bp$, $F \in \ftr \Bm$.
If $c \in F$, $F \in \ftr \Bm$ holds as $\cl^b(c) \in \Bm$.
Hence, suppose that $c \notin F$.
Let $K_F$ be the unique minimal spanning set of $F$ in $(\U, \cs^b)$.
Note that $K_F$ is indeed unique since $(\cs^b, \subseteq)$ is distributive.
Now, for each $M \in \Bp$, $\cl^b(K_F)= F \nsubseteq M = \cl^b(M)$ holds if and only if $K_F \nsubseteq M$ since $\cl^b$ is a closure operator.
We deduce that $\cl(K_F) \nsubseteq M$ for each $M \in \Mi(\cs)$ such that $c \upp M$ by definition of $\Bp$.
By definition of $\upp$, every closed set $F$ of $\cs$ either contains $c$ or satisfies $F \subseteq M$ for some $M$ such that $c \upp M$.
It follows that $c \in \cl(K_F)$.
As $c \notin K_F$, we obtain that $K_F$ contains a minimal generator of $c$.
Thus, $K_F$ can be modified (reduced) to obtain a $D$-generator $A$ of $c$ and $\cl^b(A) \subseteq F$ holds.
We obtain $F \in \ftr \Bm$ as expected.
Therefore, $\ftr \Bm \cup \idl \Bp = \cs^b$ holds.

The two antichains $\Bm$ and $\Bp$ satisfy $\ftr \Bm \cap \idl \Bp = \emptyset$ and $\ftr \Bm \cup \idl \Bp = \cs^b$.
We deduce that they are dual in $(\cs^b, \subseteq)$, which concludes this part of the proof.

The fact that $\cl^b$ defines a one-to-one correspondence between $D$-generators follows from the fact that each $A$ in $\genD(c)$ is the unique minimal spanning set of $\cl^b(A)$ in $(\cs^b, \subseteq)$.
\end{proof}

\begin{exam} \label{ex:Dual-D-base}
Let $\U = \{1, 2, 3, 4, 5\}$ and consider the closure system $(\U, \cs)$ with $D$-base $(\U, \is_D)$, where
\[ 
\is_D = \{3 \imp 2, 4 \imp 2\} \cup \left\{
\begin{array}{l l l}
25 \imp 1, & 34 \imp 1, & 15 \imp 2, \\
14 \imp 3, & 15 \imp 3, & 25 \imp 3, \\
13 \imp 4, & 25 \imp 4, & 15 \imp 4, \\
34 \imp 5, & 13 \imp 5, & 14 \imp 5
\end{array}
\right\}
\]
We have $\is^b= \{3 \imp 2, 4 \imp 2\}$.
The lattice $(\U, \cs)$ is given in Figure~\ref{fig:Dual-D-base-lattice}.
The meet-irreducible elements are $\Mi(\cs) = \{1, 12, 23, 24, 5\}$.
\begin{figure}[ht!]
    \centering
    \includegraphics[scale=1, page=1]{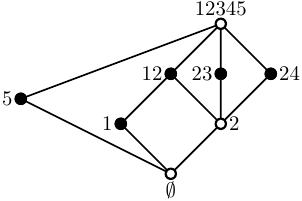}%
    \caption{The closure lattice of Example~\ref{ex:Dual-D-base}. 
    Black dots are meet-irreducible elements.}
    \label{fig:Dual-D-base-lattice}
\end{figure}
For instance, we have $\genD(5) = \{13, 14, 34\}$.
In Figure~\ref{fig:Dual-D-base-distributive} we show the distributive closure lattice $(\cs^b, \subseteq)$ associated to $(\U, \cs)$ and we highlight the two dual antichains associated to $5$:
\begin{itemize}
    \item $\Bp = \{M \in \Mi(\cs) \st 5 \upp M\} = \{12, 23, 24\}$
    \item $\Bm = \cl^b(5) \cup \{\cl^b(A) \st A \in \genD(5)\} = \{5, 123, 124, 234\}$
\end{itemize}
Since $(\cs^b, \subseteq)$ is distributive, each closed set has a unique minimal spanning set.
Moreover, $D$-generators are minimal spanning sets of their closure.
Hence, $\cl$ is a one-to-one correspondence between $\genD(5)$ and $\Bm \setminus \{\cl^b(5)\} = \{123, 124, 234\}$ as specified by Lemma~\ref{lem:gen-D-dual}.
\begin{figure}[ht!]
    \centering
    \includegraphics[page=2, scale=1]{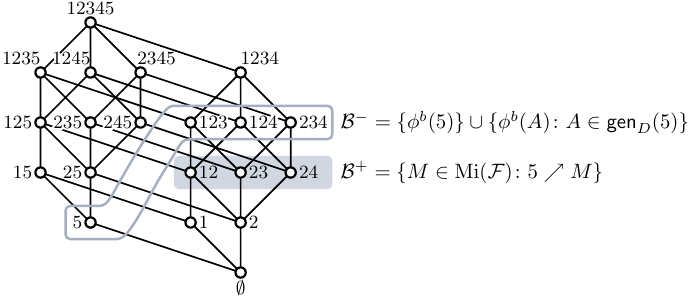}%
    \caption{The distributive closure lattice $(\cs^b, \subseteq)$ of Example~\ref{ex:Dual-D-base}.
    Two dual antichains are highlighted: the meet-irreducible elements $M$ satisfying $5 \upp M$ ($\Bp$) and the closures of $D$-generators of $5$ along with the closure of $5$ ($\Bm$).}
    \label{fig:Dual-D-base-distributive}
\end{figure}
\end{exam}

Using Lemma~\ref{lem:gen-D-dual}, we demonstrate how dualization in distributive lattices can be used to find the $D$-base from meet-irreducible elements.

\begin{lem} \label{lem:Dual-harder-D-base}
There is an output-polynomial time algorithm for computing the $D$-base of a closure system given by its meet-irreducible elements if there is one for dualization in distributive lattices.
\end{lem}

\begin{proof}
As a preliminary step, we compute the binary part $\is^b$ of $\is_D$ using $\Mi(\cs)$.
Since $\is^b = \{a \imp c \st a, c \in \U, c \in \cl(a), a \neq c\}$ and $\cl(A)$ can be computed in polynomial time from $\Mi(\cs)$ for every $A \subseteq \U$, determining $\is^b$ from $\Mi(\cs)$ takes polynomial time.
The IB $(\U, \is^b)$ is an IB for $(\U, \cs^b)$.

Suppose there exists an output-polynomial time algorithm $\ctt{A}$ for dualization of distributive lattices.
Let $c \in \U$ and let $\Bp = \{M \in \Mi(\cs) \st c \upp M\}$.
Observe that $\Bp$ can be identified in polynomial time from $\Mi(\cs)$ as $c \upp M$ is equivalent to $M \in \max_{\subseteq}(\{M \in \Mi(\cs) \st c \notin M\})$.
Now we use \ctt{A} to compute the antichain $\Bm$ dual to $\Bp$ in $(\cs^b, \subseteq)$ in output-polynomial time.
By Lemma~\ref{lem:gen-D-dual}, $\Bm = \cl^b(c) \cup \{\cl^b(A) \st A \in \genD(c)\}$.
Instead of outputting $\cl^b(A)$ for each $A \in \genD(c)$, we output $A \imp c$.
Since  $(\cs^b, \subseteq)$ is a (standard) distributive closure lattice, $A$ is the unique minimal spanning set of $\cl^b(A)$.
Henceforth, $A$ can be computed in polynomial time by greedily removing elements.
By Lemma~\ref{lem:gen-D-dual}, $\cl^b$ defines a bijection between $\genD(c)$ and $\Bm \setminus \{\cl^b(c)\}$ so that the size of $\Bm$ is bounded by $\card{\genD(c)} + 1$.
Consequently, the subset $\{A \imp c \st A \in \genD(c)\}$ is obtained in output-polynomial time using \ctt{A}.

Applying this algorithm to each element of $\U$ yields the $D$-base of the closure system.
Note that since there are only $\card{\U}$ calls to \ctt{A}, each generator will be output at most $\card{\U}$ times.
Thus, the algorithm correctly outputs the $D$-base associated to $\Mi(\cs)$ in output-polynomial time as required.
\end{proof}

Combining Lemmas~\ref{lem:D-base-harder-dual} and~\ref{lem:Dual-harder-D-base} we finally obtain Theorem~\ref{thm:D-base-dual} which we restate below.
\DBDual*

It is shown in \cite{elbassioni2022dualization} that dualization in distributive lattices can be solved in output-quasi-polynomial time.
Therefore, we obtain Theorem~\ref{thm:D-base-Mi}%Using the results in \cite{elbassioni2022dualization}, we obtain the subsequent corollary.

\DBMi*

\section{Finding the \mtt{$D$}{D}-base from an implicational base} \label{sec:D-base-IB}

In this part, we give an algorithm which computes with polynomial delay the $D$-base of a closure system given by an IB, i.e., Theorem~\ref{thm:D-base-IB}.
For this purpose we use results of Ennaoui and Nourine \cite{ennaoui2025polynomial} based on the solution graph traversal method.
The principle of this technique is to traverse a (directed) graph, called the \emph{solution graph}, whose vertices are the solutions to enumerate and whose arcs are devised from a suitably chosen \emph{transition function}.
Recall that a directed graph is \emph{strongly connected} if each vertex can be reached from any other vertex by following arcs of the graph.
The next theorem restates the conditions needed to obtain polynomial delay with solution graph traversal.
We note that to avoid repetitions during enumeration, the algorithm may require exponential space.
We redirect the reader to, e.g., \cite{boros2004algorithms, elbassioni2015polynomial, johnson1988generating} for further details on this folklore technique.

\begin{thm}[Folklore]\label{thm:solution-graph}
Let $\Pi$ be an enumeration problem which on input $x$ has solutions $\cc{S}(x)$.
We further assume that $\cc{S}(x)$ is a collection of subsets of a groundset $\U$ being part of the input $x$.
Let $\cc{N} : \cc{S}(x) \to \pow{\cc{S}(x)}$ be a transition function and define $\cc{G}$ as the directed graph---the solution graph---on $\cc{S}(x)$ where there is an arc from $S_1$ to $S_2$ if $S_2 \in \cc{N}(S_1)$. 
Then, $\cc{S}(x)$ can be computed with polynomial delay if the next three conditions hold:
\begin{enumerate}
    \item a first solution $S_0$ of $\cc{S}(x)$ can be computed in input polynomial time (i.e., polynomial in the size of $x$);
    \item for each $S \in \cc{S}$, $\cc{N}(S)$ can be computed in input polynomial time;
    \item $\cc{G}$ is strongly connected.
\end{enumerate}
\end{thm}

Let $(\U, \is)$ be an IB with closure system $(\U, \cs)$.
Our aim is to find an appropriate transition function to list all the $D$-generators of $(\U, \is)$ based on Theorem~\ref{thm:solution-graph}.
As an intermediate step, we will show how to list the $D$-generators of some element $c$ with the same technique.
For convenience, let us define $\U_c = \{a \in \U \st c \notin \cl(a)\}$.
Observe that $c \notin \U_c$ and that $\U_c$ contains at least all the elements of $\U$ that appear in $D$-generators of $c$.

\begin{rem} \label{rem:Xc}
Recall from Section~\ref{sec:preliminaries} that $\U_c \in \cs$ if and only if $c$ has no $D$-generators.
In this case, there is nothing to enumerate.
In what follows, we assume that $c$ has at least one $D$-generator to be found, hence that $\U_c \notin \cs$ and $c \in \cl(\U_c)$.
\end{rem}

We first relate the problem of finding $\gen_D(c)$ to the task of finding the so called \emph{$D$-minimal keys} of a (standard) closure system.
recall that a minimal key of a closure system $(\U, \cs)$ is an inclusion-wise minimal subset $K$ of $\U$ that satisfies $\cl(K) = \U$.
Similarly to $D$-generators, a minimal key $K$ of a closure system $(\U, \cs)$ is a \emph{$D$-minimal key} if for every minimal key $K'$, $\cl^b(K') \subseteq \cl^b(K)$ entails $K = K'$.
Using $\U_c$, we build an implicational base $(\U_c, \is_c)$ where $\is_c = \is_1 \cup \is_2$ with:
\begin{align*}
    \is_1 = & \{A \imp b \in \is \st A \cup \{b\} \subseteq \U_c \} \\ 
    \is_2 = & \{A \imp b \st A \imp d \in \is, A \subseteq \U_c, d \notin \U_c, b \in \U_c \setminus \cl^b(A)\}
\end{align*}

Note that $\is_1 \subseteq \is$ but $\is_2$ does not need to contain only valid implications of $(\U, \cs)$.
Moreover, observe that $(\U_c, \is_c)$ is standard.
Indeed, $\emptyset$ remains closed in $\cs_c$ as long as it is already in $\cs$, and $\cl(a) = \cl_c(a)$ by definition of $\is_c$ as long as $a \in \U_c$. 
More generally, $\is_c^b = \is^b \setminus \{a \imp d \st c \in \cl(d) \}$, so that for every $D$-generator $A$ of $c$, $\cl_c^b(A) = \cl^b(A)$.
Furthermore, $(\U_c, \is_c)$ can be computed in polynomial time in the size of $(\U, \is)$.
The reduction is illustrated in Example~\ref{ex:solution-graph}.
We argue that $D$-generators of $c$ in $(\U, \cs)$ are precisely the $D$-minimal keys of $(\U_c, \cs_c)$.

\begin{lem} \label{lem:link-minimal-D}
The $D$-generators of $c$ in $(\U, \is)$ are precisely the $D$-minimal keys of $(\U_c, \is_c)$.
\end{lem}

\begin{proof}
As a preliminary observation, remark that both $D$-minimal keys of $(\U_c, \cs_c)$ and $D$-generators of $c$ lie in $\U_c$ by definition.

We prove that for every $S \subseteq \U_c$, $\cl_c(S) = \U_c$ if and only if $c \in \cl(S)$.
Suppose first that $c \in \cl(S)$ and let $S = S_0, \dots, S_k = \cl(S)$ be the sequence of sets obtained using forward chaining on $S$ with $\is$.
Since $S \subseteq \U_c$ and $c \in \cl(S)$, there is some $1 \leq i \leq k$ and some implication $A \imp d \in \is$ such that $A \subseteq \U_c$, $A \subseteq S_{i - 1}$, $d \notin \U_c$ and $d \notin S_{i - 1}$.
Consider the least such $i$.
By assumption, $S_0, S_1, \dots, S_{i-1}$ are included in $\U_c$, and hence the implications used to build these sets belong to $\is_1$.
Thus, they will be used in the forward chaining on $S$ with $\is_c$ too so that $A \subseteq \cl_c(S)$ holds.
Because $A \imp b \in \is_2$ for every $b \in \U_c \setminus \cl^b(A)$, we deduce that $\cl_c(S) = \U_c$ as expected.

Now assume that $\cl_c(S) = \U_c$.
By Remark~\ref{rem:Xc}, $c \in \cl(\U_c)$.
Again, consider the forward chaining on $S$ with $\is_c$.
In order to obtain $\U_c$, we have two cases:
\begin{enumerate}[(1)]
    \item an implication $A \imp b$ of $\is_2$ is used.
    It corresponds to an implication $A \imp d \in \is$ such that $d \notin \U_c$.
    Consider the first such implication encountered in the forward chaining.
    Then, all implications used before $A \imp b$ belong to $\is$.
    In particular, $A \imp d$ will be used in when applying forward chaining on $S$ with $\is$, so that $d \in \cl(S)$.
    Since $d \notin \U_c$ entails $c \in \cl(d)$, we deduce $c \in \cl(S)$.
    
    \item every implication used in the forward chaining, if any, belongs to $\is_1$ (hence to $\is$).
    In this case, $\U_c \subseteq \cl_c(S)$ implies $\U_c \subseteq \cl(S)$.
    This holds in particular if $S = \U_c$, which covers the case where $\is_c = \emptyset$.
    Since $\U_c$ is assumed not to be closed in $\cs$, it contains at least one $D$-generator of $c$.
    We deduce $c \in \cl(S)$ as expected.
\end{enumerate}
This concludes the proof that $c \in \cl(S)$ if and only if $\cl_c(S) = \U_c$.
In particular, it holds that $S$ is a minimal generator of $c$ included in $\U_c$ if and only if $S$ is a minimal key of $(\U_c, \cs_c)$.
As the binary part of $\is$ involved in $D$-generators of $c$ is included in $\is_c$, it further holds that $S$ is a $D$-generator of $c$ if and only if it is a $D$-minimal key of $(\U_c, \cs_c)$.
\end{proof}

In \cite{ennaoui2025polynomial}, the authors show how to list the $D$-minimal keys of a standard closure system using solution graph traversal.
As a consequence of Lemma~\ref{lem:link-minimal-D}, we can adapt their approach to the problem of finding $\genD(c)$.
In particular, we rephrase in our framework a key statement that lays the ground for defining a suitable transition function.
For self-containment, we give a proof of this statement in our setup in \ref{app:transition}.

\begin{lem}[restate=TRA, label=lem:transition]{\cite[Theorem 8]{ennaoui2025polynomial}}
Let $(\U, \is)$ be a standard IB, let $c \in \U$ and let $\cc{S}$ be a non-empty subset of $\genD(c)$.
Then, $\cc{S} \neq \gen_D(c)$ if and only if there exists some $A$ in $\cc{S}$ and an implication $B \imp d \in \is_c$, $\card{B} \geq 2$, such that $\cl_c^b((\cl_c^b(A) \setminus \cl_c^b(d)) \cup B)$ does not contain any set of $\cc{S}$.
\end{lem}

In words, Lemma~\ref{lem:transition} states that the set $\cc{S}$ is incomplete if and only if there is a closed set in $\cs_c^b$ which includes an unseen $D$-generator.
Indeed, based on Lemma~\ref{lem:gen-D-dual}, $\{\cl_c^b(A) \st A \in \genD(c) \}$ defines an antichain of $(\cs_c^b, \subseteq)$.
If the antichain $\{\cl_c^b(A) \st A \in \cc{S}\}$ associated to $\cc{S}$ is incomplete, then there is a closed set $F \in \cs_c^b$ which is in $\ftr \{\cl_c^b(A) \st A \in \genD(c) \}$ but not in $\ftr \{\cl_c^b(A) \st A \in \cc{S}\}$.
The lemma further shows that such a closed set can be identified by substituting elements from the $\cl_c^b$-closure of a solution in $\cc{S}$ for a premise in $\is_c$.
Henceforth, in view of Theorem~\ref{thm:solution-graph} and according to Lemma~\ref{lem:transition}, we define the solution graph $\cc{G}(c)$ with vertices $\genD(c)$ and transition function $\cc{N}$ defined by:
\[ 
\cc{N}(A) = \left\{ \ctt{Min}(\cl_c^b((\cl_c^b(A) \setminus \cl_c^b(d)) \cup B)) \st B \imp d \in \is_c, \card{B} \geq 2\right\}.
\]
Here, $\ctt{Min}$ is a greedy algorithm that computes a $D$-generator (of $c$) from any $\cl_c^b$-closed generator of $c$.
To do so, it removes the first extreme element of the current closed set (recall that $(\U_c, \cs_c^b)$ is a convex geometry) such that the resulting closed set still subsumes a $D$-generator of $c$.
To determine which extreme point is the first, we use a linear extension on $X_c$ of the natural order induced by $\cl^b_c$.

When the algorithm reaches a closed set of $\cs_c^b$ whose extreme elements satisfy Lemma~\ref{lem:D-carac}, it outputs the corresponding $D$-generator.
The whole procedure can be conducted in polynomial time given $(\U, \is)$ and $(\U_c, \is_c)$.

\begin{rem} \label{rem:order}
Since for every $c \in \U$, $\is_c^b \subseteq \is^b$, we can use the same ordering of $\U$ which is then restricted to each $\U_c$.
\end{rem}

Now, observe that $\cc{G}(c)$ is strongly connected.
Let $\cc{S}$ be a non-empty subset of $\genD(c)$.
Based on the definition of $\cc{N}$, Lemma~\ref{lem:link-minimal-D} states that either $\cc{S} = \genD(c)$ or $\cc{S}$ has an unseen neighbor w.r.t.~$\cc{N}$.
Henceforth, the process of starting from $\cc{S} = \{A\}$ for some $D$-generator $A$ of $c$ and repeatedly applying Lemma~\ref{lem:link-minimal-D} to obtain a new neighbor of $S$ will eventually reach each generator in $\genD(c)$.
As we can use this strategy starting from any generator, we deduce that $\cc{G}(c)$ is strongly connected.
Moreover, a first $D$-generator can be computed in polynomial time using $\ctt{Min}$ on $\U_c$.
Finally $\cc{N}(A)$ can be generated in polynomial time.
Hence, all the requirements of Theorem~\ref{thm:solution-graph} are satisfied.
Computing $(\U_c, \is_c)$ in polynomial time at pre-processing, we obtain:

\begin{thm} \label{thm:gen-c-polydelay}
Let $(\U, \is)$ be an implicational base of a standard closure system $(\U, \cs)$ and let $c \in \U$.
Then, there is a polynomial-delay algorithm which computes $\genD(c)$ from $(\U, \is)$. 
\end{thm}

\begin{exam} \label{ex:solution-graph}
Let $\U = \{1, 2, 3, 4, 5, 6, 7, 8\}$ and consider the IB $(\U, \is)$ defined by:
\[ 
\is = 
\left\{ 
\begin{array}{l}
4 \imp 3, \\ 
3 \imp 2, \\
2 \imp 1
\end{array} \right\} \cup 
\left\{
\begin{array}{l l}
15 \imp 2, & 16 \imp 2, \\
27 \imp 3, & 28 \imp 3, \\
36 \imp 4, & 37 \imp 4
\end{array}
\right\}
\]
The corresponding closure system $(\U, \cs)$ is standard.
The elements of $\U$ admitting $D$-generators are $2$, $3$ and $4$:
\begin{align*}
\genD(2) = & \; \{15, 16\}, \\ 
\genD(3) = & \; \{157, 158, 167, 168, 27, 28\}, \\
\genD(4) = & \; \{157, 167, 168, 27, 36\}.
\end{align*}
We consider the element $4$.
We have $\U_4 = \U \setminus \{4\}$ and:
\[ 
\is_4 = 
\left\{ 
\begin{array}{l} 
3 \imp 2, \\
2 \imp 1
\end{array} \right\} \cup 
\left\{
\begin{array}{l l}
15 \imp 2, & 16 \imp 2, \\
27 \imp 3, & 28 \imp 3 \\
\end{array}
\right\}
\cup
\left\{ 
\begin{array}{l l l}
36 \imp 5, & 36 \imp 7, & 36 \imp 8, \\
37 \imp 5, & 37 \imp 6, & 37 \imp 8
\end{array}
\right\}.
\]
The solution graph $\cc{G}(4)$ is given in Figure~\ref{fig:solution-graph}.
We use the ordering $1 < 2 < \dots < 8$ for the \ctt{Min} procedure.
Remark that it complies with $\cl_4^b$.
For instance, there is a transition from $167$ to $168$, i.e., $168 \in \cc{N}(167)$, given by the implication $28 \to 3$ (highlighted in the figure).
We have that $167$ is a $D$-generator of $4$, so that $\cl_4^b((\cl_4^b(167) \setminus \cl_4^b(3)) \cup 28) = 12678$ must contain a $D$-generator of $4$ as $1 \in \cl_4^b(3)$. 
The \ctt{Min} procedure will remove $2$ and $7$ in that order.
We obtain $168$ as a new $D$-generator (since $168$ is at the same time a $D$-generator and a closed w.r.t.~ $\cl_4^b$).
\begin{figure}[ht!]
    \centering
    \includegraphics[page=1, scale=0.9]{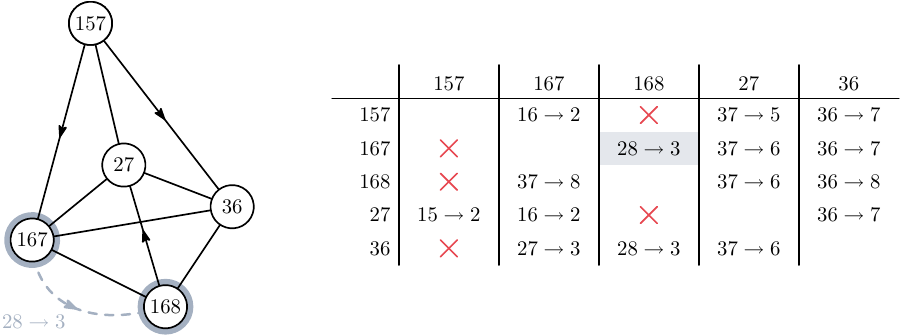}
    \caption{On the left, the solution graph $\cc{G}(4)$ of Example~\ref{ex:solution-graph}. An edge without arrows indicate transitions in both directions. On the right, a table that gives an example of implication that can be used for each transition. For instance, the implication $28 \to 3$ yields a transition from $167$ to $168$ (highlighted in grey). A cross means that the transition is not possible. There may be other possible choices of implications, and self-loops are omitted for simplicity.}
    \label{fig:solution-graph}
\end{figure}
We represent the three solution graphs $\cc{G}(2), \cc{G}(3)$, $\cc{G}(4)$ in Figure~\ref{fig:solution-graph-2}.
We note that the ordering used for $\U_2$, $\U_3$ is a restriction of the same ordering on $X$ that was used for $\U_4$, per Remark~\ref{rem:order}.

\begin{figure}[ht!]
    \centering
    \includegraphics[page=2, scale=0.9]{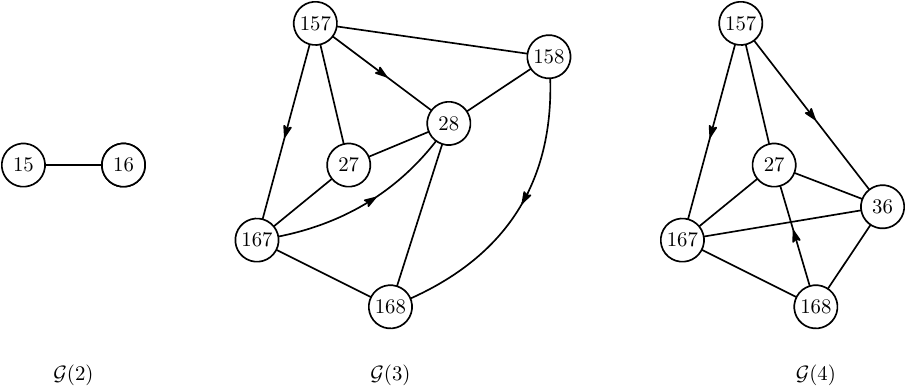}
    \caption{The three solution graphs of Example~\ref{ex:solution-graph}.}
    \label{fig:solution-graph-2}
\end{figure}
\end{exam}

We build upon previous discussions to give an algorithm finding all the $D$-generators of $(\U, \is)$ and outputting the corresponding implications of the $D$-base.
The strategy of sequentially applying Theorem~\ref{thm:gen-c-polydelay} to each $x \in \U$ may not guarantee polynomial delay as a same set $A$ may be a $D$-generator of several distinct elements.
This is the case in Example~\ref{ex:solution-graph} where we have $\genD(3) \cap \genD(4) = \{157, 167, 168, 27\}$.
Instead, we will use the graphs $\cc{G}(x)$ to build a new solution graph $\cc{G}(\U)$ on which we will be able to apply Theorem~\ref{thm:solution-graph}.
The graph $\cc{G}(\U)$ has vertices $\bigcup_{c \in \U} \genD(c)$ and its transition function is defined as:
\[
\cc{N}(A) = \bigcup_{\substack{c \in \U, \\ A \in \genD(c)}} \left\{ \ctt{Min}(\cl_c^b((\cl_c^b(A) \setminus \cl_c^b(d)) \cup B)) \st B \imp d \in \is_c, \card{B} \geq 2\right\}
\]
In other words, $\cc{G}(\U) = \bigcup_{c \in \U} \cc{G}(c)$.
In Figure~\ref{fig:solution-graph-3}, we show the solution graph $\cc{G}(\U)$ associated with Example~\ref{ex:solution-graph}.

\begin{figure}[ht!]
    \centering
    \includegraphics[page=3, scale=0.9]{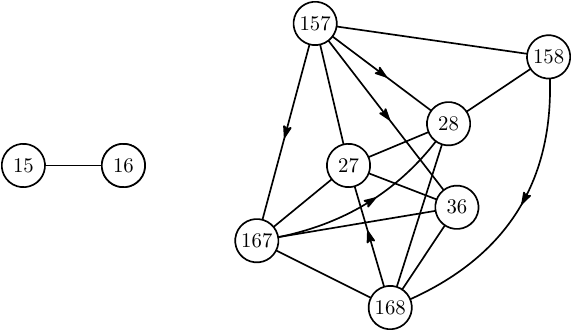}
    \caption{The solution graph $\cc{G}(\U)$ of Example~\ref{ex:solution-graph}.
    It the union of two strongly connected components: one associated to $\cc{G}(2)$ (on the left), and one associated to $\genD(3)$ and $\genD(4)$ (on the right).}
    \label{fig:solution-graph-3}
\end{figure}

The graph $\cc{G}(\U)$ is an union of strongly connected components by Lemma~\ref{lem:transition}.
The transition function can be computed in polynomial time as it is the union of polynomial-time computable transition functions.
Finally we can compute in polynomial time a $D$-generator for each strongly connected component.
Henceforth, Theorem~\ref{thm:solution-graph} applies to each such component of $\cc{G}(\U)$.
This gives an algorithm which consists in traversing each strongly connected component one after the other.
At the moment of outputting a set $A$, the algorithm outputs the implications $A \imp c$ such that $A \in \genD(c)$.
We thus obtain Theorem~\ref{thm:D-base-IB}:

\DBIB*

\begin{proof}
As a pre-processing step, we compute $(\U, \is^b)$ and $(\U_c, \is_c)$ for each $c$ such that $\genD(c) \neq \emptyset$.
This can be done in polynomial time (see Section~\ref{sec:preliminaries}).
Then, we apply Theorem~\ref{thm:solution-graph} on each strongly connected component of $\cc{G}$ one after the other.
At the end of the traversal of a strongly connected component, we have to detect that there is no component left.
On this purpose, it is sufficient to test whether a $D$-generator have been output for each $c$ such that $\genD(c) \neq \emptyset$.
If this is the case, the $D$-base has been found completely.
Otherwise, we use $\ctt{Min}(\U_c)$ on a $c$ whose $D$-generators have not been enumerated to start a new Breadth-First Search in one of the remaining components.
\end{proof}

We remark that, much as Theorem~\ref{thm:gen-c-polydelay}, the algorithm of Theorem~\ref{thm:D-base-IB} may require an exponential amount of space, notably to avoid repetitions.
The question of whether the same time complexity can be achieved with polynomial space is intriguing.
However, answering this question in the affirmative would at the same time give an algorithm to list all the minimal keys of a set of functional dependencies with polynomial delay and polynomial space, a long-standing open question \cite{ennaoui2025polynomial, lucchesi1978candidate}.

\section{Conclusion} \label{sec:conclusion}

In this paper, we have investigated complexity aspects of the $D$-base and $D$-relation of a closure system.
We have shown that computing the $D$-relation from an IB is $\NP$-complete even in the acyclic setup.
Besides, we gave output sensitive algorithms to compute the $D$-base from both an IB and meet-irreducible elements.
Namely, we gave a polynomial-delay algorithm to find the $D$-base from an arbitrary IB, and an output-quasi-polynomial time algorithm for recovering the $D$-base from meet-irreducible elements.
These results answer Questions~\ref{ques:D-relation},~\ref{ques:D-base-Mi} and~\ref{ques:D-base-IB} and hence complete the complexity landscape of the $D$-relation and $D$-base.
However, the properties of the $D$-relation could be further studied.
Indeed, the fact that lower bounded closure systems are characterized by an acyclic $D$-relation \cite{freese1995free} suggests that other properties of closure systems may be embedded in the structure of the $D$-relation.
This leads to the following intriguing question:

\begin{ques}
Are there any further properties of closure systems that can be derived directly from the $D$-relation?
\end{ques}

Another future research direction concerns the $E$-base and the $E$-relation \cite{adaricheva2014implicational, adaricheva2017optimum}.
The $E$-base is a further refinement of the $D$-base, and the $E$-relation is defined in a similar way as the $D$-relation from $D$-generators.
More precisely, a $D$-generator of an element $c$ is an $E$-generator of $c$ if its closure in the underlying closure system is minimal among the closures of the other $D$-generators of $c$.
This makes the $E$-base connected with critical circuits of convex geometries \cite{korte2012greedoids}.
Unlike the $D$-base though, the $E$-base does not always constitute a valid IB of its associated closure system.
Henceforth, the most intriguing question regarding the $E$-base is:

\begin{ques}
Which closure systems have valid $E$-base?
\end{ques}

Besides, the questions regarding the complexity of finding the $E$-base or the $E$-relation given an IB or meet-irreducible elements, much as what has been done for the $D$-base, are left to be answered.

\bibliographystyle{alphaurl}
\bibliography{biblio}

\newpage 

\appendix
\section{Proof of Theorem \ref{thm:D-relation-ACG}}
\label{app:D-relation-ACG}

In this section, we give the complete proof of Theorem \ref{thm:D-relation-ACG} which we first restate.

\DACG*  % restate theorem

We showed in Corollary \ref{cor:DB-M} that the problem belongs to $\NP$.
To show $\NP$-hardness, we use a reduction from the problem $1$-in-$3$-SAT:

\begin{decproblem}
\problemtitle{1-in-3 SAT}
\probleminput{A positive $3$-CNF $\varphi = \{C_1, \dots, C_m\}$ over variables $V = \{v_1, \dots, v_n\}$.}
\problemquestion{Is there a truth assignment of the variables in $V$ which is a model of $\varphi$ and such that each clause of $\varphi$ has exactly one true literal?}
\end{decproblem}

Let $\varphi = \{C_1, \dots, C_m\}$ be a non-trivial positive $3$-CNF over $V = \{v_1, \dots, v_n\}$.
For simplicity, we view clauses as sets of variables.
Moreover, we identify an assignment of $V$ to the sets of its variables assigned 1.
Consequently, a 1-in-3 assignment of $V$ is a set of variables $T \subset V$ such that $\card{C_i \cap T} = 1$ for each clause $C_i$ of $\varphi$.
Two variables $v_i, v_j$ are in conflict if there exists a clause $C_k$ in $\varphi$ that contains both $v_i$ and $v_j$.
For each $1 \leq i \leq m$, we define $\varphi_i = \{C_i, \dots, C_m\}$.
A subset $T$ of $V$ is $\varphi$-conflict-free (resp. $\varphi_i$-conflict-free) if it does not contain variables in conflict among the clauses of $\varphi$ (resp. $\varphi_i$).

We build an instance of $D$-relation identification.
Let $C = \{c_1, \dots, c_m, c_{m + 1}\}$ and let $\U = C \cup V$.
An element $c_i$, $1 \leq i \leq m$ represents the clause $C_i$.
The element $c_{m + 1}$ is another gadget element.
Now consider the set $\is = \is_{\mathit{mod}} \cup \is_{\mathit{conf}}$ of implications, where:
\begin{align*}
    \is_{\mathit{mod}} & = \{c_i v_j \imp c_{i + 1} \st 1 \leq i \leq m \text{ and } v_j \in C_i\} \\
    \is_{\mathit{conf}} & = \{v_i v_j \imp c_{m + 1} \st v_i, v_j \in V \text{ and } v_i, v_j \text{ are in conflict in } \varphi\}
\end{align*}
Informally, $\is_{\mathit{mod}}$ models the fact that whenever the $i$-th clause is satisfied, we can proceed to the $(i+1)$-th clause or $c_{m + 1}$, and $\is_{\mathit{conf}}$ represents the conflicts induced by the variables appearing together in a clause.

We take $(\U, \is)$ as input to our problem, together with $c_1$ and $c_{m + 1}$.
Note that $(\U, \is)$ can be constructed in polynomial time in the size of $\varphi$ and $V$.
The premises of $\is$ have size $2$ and $G(\is)$ has no cycle.
Hence, $(\U, \cs)$ is an acyclic closure system represented by an IB with premises of size at most $2$ as required.

We argue that $c_1$ belongs to a $D$-generator of $c_{m + 1}$ if and only if there exists a 1-in-3 assignment $T$ of $\varphi$.
As a preliminary observation, note that $\is$ has no binary implications.
Hence, $\gen(c_{m + 1}) = \gen_D(c_{m + 1})$ and we can restrict our analysis to minimal generators of $c_{m + 1}$ only.
Remark that by definition of $\is$, $\gen(c_{m + 1}) \neq \emptyset$ and for every $A \in \gen(c_{m + 1})$, $\card{A} \geq 2$.
We start our analysis with a useful proposition.

\begin{prop} \label{prop:cascade}
For any $A \subseteq \U$ such that $A \setminus C$ is $\varphi$-conflict-free and any $1 < i \leq m +1$, $c_i \in \cl(A)$ if and only if $c_i \in A$ or $c_{i - 1} \in \cl(A)$ and $C_{i - 1} \cap A \neq \emptyset$.
\end{prop}

\begin{proof}
The if part follows from the definition of $\is$.
We prove the only if part.
Assume that $c_i \in \cl(A)$ but $c_i \notin A$.
Then, in the forward chaining $A = A_0, \dots, A_k = \cl(A)$ applied on $A$ with $\is$, there will be some $A_j$ and some implication $B \imp c_i$ such that $B \subseteq A_j$ and $c_i \notin A_j$, so that $c_i \in A_{j + 1}$. 
Since $A \setminus C$ is $\varphi$-conflict-free, $B \imp c_i$ must be of the form $c_{i - 1} v_j \imp c_i$ for some $v_j \in C_{i - 1}$.
As $B \subseteq A_j \subseteq \cl(A)$, we obtain $c_{i - 1}, v_j \in \cl(A)$.
It remains to show that $v_j \in A$ so that $A \cap C_{i - 1} \neq \emptyset$.
But this follows from the fact that no implication has $v_j$ as a conclusion.
Thus, $c_{i - 1} \in \cl(A)$ and $C_{i - 1} \cap A \neq \emptyset$.
This concludes the proof.
\end{proof}

Next, we characterize the minimal generators of $c_{m + 1}$.

\begin{prop} \label{prop:genz}
The following equality holds:
\begin{align*}
\gen(c_{m + 1}) = & \{\{v_i, v_j \} \st v_i, v_j \in V \text{ are in conflict in } \varphi\} \\
& \cup \min_{\subseteq}\{ T_i \cup \{c_i\} \st T_i \subseteq V \text{ is a $\varphi$-conflict-free 1-in-3 assignment of } \varphi_i \}
\end{align*}
\end{prop}

\begin{proof}
We show double-inclusion.
We begin with the $\supseteq$ part.
The fact that each pair $\{v_i, v_j\}$ of variables in conflict is a minimal generator of $c_{m + 1}$ follows from the definition of $\is$.
Let $T_i \subseteq V$ be an inclusion-wise minimal $\varphi$-conflict-free 1-in-3 assignment of $\varphi_i$, if it exists, and put $A = T_i \cup \{c_i\}$.
By assumption, $C_{j} \cap A \neq \emptyset$ for every $i \leq j \leq m$.
Furthermore, $c_i \in A$.
Hence using Proposition \ref{prop:cascade} inductively on $i < j \leq m + 1$, we deduce that $c_{m + 1} \in \cl(A)$ holds as required.
Now we show that no proper subset $A'$ of $A$ satisfies $c_{m + 1} \in \cl(A')$.
Consider first $A' = A \setminus \{c_i\} = T_i$.
By assumption, $T_i$ is $\varphi$-conflict-free.
Hence, $T_i$ contains no premise of $\is$ so that $\cl(T_i) = T_i$ and $c_{m + 1} \notin T_i = A'$.
Let $A' = A \setminus \{v_j\}$ for some $v_j \in T_i$.
Since $T_i$ is an inclusion-wise minimal 1-in-3 assignment of $\varphi_i$, there exists a clause $C_k$ such that $C_k \cap A = \{v_j\}$ so that $C_k \cap A' = \emptyset$, with $k \geq i$.
Again using Proposition \ref{prop:cascade} inductively on $k < \ell \leq m + 1$, we deduce that $c_{m + 1} \notin \cl(A')$.
This concludes the proof that $A \in \gen(c_{m + 1})$.

We proceed to the $\subseteq$ part.
Let $A$ be a minimal generator of $c_{m + 1}$.
We have two cases: $A \cap C = \emptyset$ and $A \cap C \neq \emptyset$.
In the first case, $c_{m + 1} \in \cl(A)$ implies that $A$ is not closed.
Hence, there exists an implication $B \imp d \in \is$ such that $B \subseteq A \subseteq V$.
As $A$ is a minimal generator of $c_{m + 1}$ by assumption, we deduce $A = \{v_i, v_j\}$ for some variables $v_1, vj$ in conflict as expected.

Now we consider the case $A \cap C \neq \emptyset$.
Since $A$ is a minimal generator of $c_{m + 1}$, $A \setminus C$ is $\varphi$-conflict-free.
Let $i$ be the maximal index in $1 \leq i \leq m$ such that $c_i \in A$.
Since $A \setminus C$ is $\varphi$-conflict-free, $c_{m + 1} \in \cl(A)$ implies $c_m \in \cl(A)$ and $C_m \cap A \neq \emptyset$ by definition of $\is$ and Proposition \ref{prop:cascade}.
Using Proposition \ref{prop:cascade} inductively for $j$ ranging from $m + 1$ to $i + 1$, we deduce that $A \cap C_k \neq \emptyset$ for every $i \leq k \leq m$, so that $A \setminus C$ is a $\varphi$-conflict-free 1-in-3 assignment of $\varphi_i$.
Using $\subseteq$ part, we further deduce that $T_i = A \setminus C$ is inclusion-wise minimal for this property.
Hence, $A$ is of the form $T_i \cup \{c_i\}$ as required, which concludes the proof.
\end{proof}

We are in position to conclude that $\varphi$ has a valid $1$-in-$3$ assignment if and only if there exists a minimal generator of $c_{m + 1}$ containing $c_1$.
Since $\is^b$ is empty, this is equivalent to show that there exists a $D$-generator of $c_{m+1}$ containing $c_1$, i.e., $c_{m + 1} D c_1$.
Assume that $\varphi$ has a valid $1$-in-$3$ assignment $T$.
Then, by Proposition \ref{prop:genz}, $T \cup \{c_1\}$ belongs to $\gen(c_{m + 1})$ and $c_{m + 1} \delta c_1$ holds.
Conversely, if there exists a minimal generator $A$ of $c_{m + 1}$ containing $c_1$, then by Proposition \ref{prop:genz}, $A \setminus \{c_1\}$ is a $1$-in-$3$ assignment of $\{C_1, \dots, C_m\} = \varphi$.
This concludes the proof of the theorem.

\section{Proof of Theorem \ref{thm:D-relation-LB}}
\label{app:D-relation-LB}

We proceed to the proof of Theorem \ref{thm:D-relation-LB} which we first recall.

\DLB*

We showed in Corollary \ref{cor:DB-M} that $D$-relation identification belongs to $\NP$.
To show $\NP$-hardness, we use another reduction from $1$-in-$3$ SAT.
Let $\varphi = \{C_1, \dots, C_m\}$ be a positive $3$-CNF over variables $V = \{v_1, \dots, v_n\}$.
Let $C = \{c_1, \dots, c_m\}$.
We use the same terminology as in Theorem \ref{thm:D-relation-ACG}.
Let $\U = V \cup C \cup \{a, b\}$ be a new groundset, and let $\is = \is_{\mathit{mod}} \cup \is_{\mathit{conf}} \cup \is^b$ where:
\begin{align*}
\is_{\mathit{mod}} & = \{a v_i \imp c_j \st v_i \in C_j \} \cup \{C \imp b\} \\
\is_{\mathit{conf}} & = \{v_i v_j \imp b \st v_i, v_j \in V \text{ and } v_i, v_j \text{ are in conflict in } \varphi \} \\
\is^b & = \{c_i \imp a \st c_i \in C\}
\end{align*}
Informally, $\is_{\mathit{conf}}$ models all the conflicts in between two variables appearing in a clause and $\is_{\mathit{mod}}$ models the fact that taking one element in a clause together with $a$ for all the clauses yields a $1$-in-$3$ assignment of $\varphi$.
The binary implications $\is^b$ guarantee that $a$ will never appear in any minimal generator together with one of the $c_i$'s.
Finally, let $(\U, \cs)$ be the closure system associated to $\is$.
We call $\cl$ the corresponding closure operator.
The reduction can be conducted in polynomial time.

First, we show in the subsequent statement that the closure system $(\U, \cs)$ associated to $(\U, \is)$ is standard and without $D$-cycles.

\begin{prop} \label{prop:std}
The closure system $(\U, \cs)$ is standard and without $D$-cycles.
\end{prop}

\begin{proof}
Recall that a closure system $(\U, \cs)$ is standard if and only if $\cl(e) \setminus \{e\} \in \cs$ for each $e \in \U$.
For each $e \in \U \setminus C$, we have $\cl(e) = \{e\}$ and $\cl(e) \setminus \{e\} = \emptyset \in \cs$ by definition of $\is$.
For each $c_i \in C$, $\cl(c_i) = \{c_i, a\}$.
Since $\cl(a) = \{a\}$, $\cl(c_i) \setminus \{c_i\} \in \cs$ follows.
We conclude that $(\U, \cs)$ is standard.

We demonstrate that the $D$-relation has no cycles.
For $e \in \U$, we denote by $D(e)$ the elements $f$ of $\U$ such that $e D f$.
These elements are the $D$-neighbors of $e$.
We have the following properties:
\begin{enumerate}
    \item For every $e \in \U \setminus \{b\}$, $b \notin D(e)$.
    Let $A \subseteq \U$ and assume that $b \in A$, $e \notin A$ and $e \in \cl(A)$.
    Then, there is a sequence of implications in $\is$ used in the forward chaining applied to $A$ in $\is$ that allows to derive $e$.
    However, $b$ belongs to no premise of $\is$.
    Hence $e \in \cl(A \setminus \{b\})$.
    We deduce that $b$ belongs to no minimal generator of $e$, and that $b \notin D(e)$.
    
    \item $D(a) = \emptyset$ and $D(v_j) = \emptyset$ for every $v_j \in V$.
    Let $v_j \in V$.
    By definition of $\is$, no implication has $v_j$ has a conclusion.
    Hence $\U \setminus \{v_j\}$ is closed, so that $\gen(v_j) = \emptyset$ and $D(v_j) = \emptyset$.
    Now consider $a$.
    Again by definition of $\is$, $\U \setminus (C \cup \{a\})$ is closed.
    Hence, any minimal generator of $a$ must intersect $C$.
    However, for every $c_j  \in C$, we have $c_j \imp a \in \is$ due to $\is^b$.
    We conclude that $\gen(a) = C$ and hence $D(a) = \emptyset$.
    
    \item For every $c_j \in C$, $D(c_j) \subseteq \U \setminus C$.
    Assume for contradiction there exists $c_j \in C$ such that $D(c_j) \cap C \neq \emptyset$.
    In order to derive $c_j$ from $A$ in $\is$, the forward chaining procedure must use an implication of the form $a v_i \imp c_j$ where $v_i \in C_j$, that is $a, v_i \in \cl(A)$.
    Since no implication of $\is$ has $v_i$ as a conclusion, we also have $v_i \in A$.
    Now, let $c_k \in A \cap C$.
    Such a $c_k$ exists by assumption.
    We have $\{a\} = \cl^b(a) \subset \cl^b(c_k) \subseteq \cl^b(A)$.
    Since $\cl^b(v_i) = \{v_i\}$, we obtain $\cl^b(\{a, v_i\}) \subset \cl^b(A)$.
    However, $c_j \in \cl^b(\{a, v_i\})$.
    This contradicts $A$ being a $D$-generator of $c_j$.
    We deduce that $D(c_j) \subseteq \U \setminus C$ for all $c_j \in C$.
\end{enumerate}
Using items 1. and 2., we deduce that no cycle in the $D$-relation can pass through any elements of $V \cup \{a, b\}$.
In other words, any $D$-cycle, if any, must be included in $C$.
However, by item 3., for every $c_j \in C$, $D(C_j) \cap C = \emptyset$ so that no $D$-cycle can contain $2$ adjacent elements of $C$.
We must conclude that the $D$-relation has no cycle as required.
\end{proof}

It remains to show that $bDa$ holds if and only if $\varphi$ admits a valid 1-in-3 assignment.
Recall that $bDa$ holds exactly if there exists a $D$-generator $A$ of $b$ such that $a \in A$.
Since $c_j \imp a \in \is$ for each $c_j \in C$, it follows that if such a $A$ exists, it satisfies $A \subseteq V \cup \{a\}$.
Moreover, $\cl(d) = \cl^b(d) = \{d\}$ for each $d \in V \cup \{a\}$.
Henceforth, any minimal generator of $b$ included in $V \cup \{a\}$ is already $\cl^b$-minimal and hence a $D$-generator of $b$.
Thus, $bDa$ is equivalent to $a$ belonging to a minimal generator of $b$.

We prove the if part of the statement.
Assume that $bDa$ holds.
We show that $\varphi$ has a valid 1-in-3 assignment $T$.
Following previous discussion, let $A$ be a minimal generator of $b$ with $a \in A$ and let $T = A \cap V$.
Because $A$ contains $a$ and $A$ is a minimal generator of $b$, $b \notin \cl(A \setminus \{a\})$.
By construction of $\is_{\mathit{conf}}$ we deduce that $T$ is $\varphi$-conflict-free.
Thus, for every $C_i \in \varphi$, $\card{T \cap C_i} \leq 1$.
On the other hand, $b \in \cl(A)$.
Since $\cl(A) \cap V = A \cap V = T$, we deduce that in $\is$, $b$ is derived from $A$ with the implication $C \imp b$.
Again using $\is$ and the fact that $A \subseteq V \cup \{a\}$, we deduce that $T$ must contain at least one variable for each clause $C_i$ of $\varphi$.
Consequently, $\card{T \cap C_i} = 1$ holds for each $C_i \in \varphi$, and $T$ is a valid 1-in-3 assignment of $\varphi$.

We move to the only if part.
Let $T$ be an inclusion-wise minimal valid 1-in-3 assignment of $\varphi$.
We show that $T \cup \{a\}$ is a minimal generator of $b$.
Since $T$ contains one variable for each clause $C_j$ in $\varphi$, $c_j \in \cl(T \cup \{a\})$ holds for every $c_j \in C$ by definition of $\is$.
Since $C \imp b \in \is$, we have $b \in \cl(T \cup \{a\})$.
We show that no proper subset $T'$ of $T \cup \{a\}$ implies $b$.
If $T' = T \setminus \{a\}$, then $\cl(T') = T' = T$ as $T$ is $\varphi$-conflict-free.
Let $T' = T \cup \{a\} \setminus \{v_i\}$ for some $v_i \in T$.
Since $T$ is inclusion-wise minimal, there exists $C_j \in \varphi$ such that $T \cap C_j = \{v_j\}$.
Thus, $T' \cap C_j = \emptyset$.
It follows that $c_i \notin \cl(T')$.
As $T$ is $\varphi$-conflict-free, it contains no premise of $\is_{\mathit{conf}}$.
We deduce that $b \notin \cl(T')$.
Applying this reasoning to each $v_j \in T \cap V$, we conclude that that $T \cup \{a\}$ is a minimal generator of $b$.
By previous discussion, we obtain $bDa$.
This concludes the proof of the theorem.

\section{Proof of Proposition \ref{prop:D-base-dual-DB}} 
\label{app:D-base-dual-DB}

In this section we prove Proposition \ref{prop:D-base-dual-DB} which we first restate.

\PDB*  % repeat proposition

\begin{proof}
We first show that $\is'^b = \is^b \cup \{a \imp d \st B \subseteq \cl(a) \text{ for some } B \in \Bm \}$ holds.
We begin with the $\supseteq$ part.
Let $a \imp b$ be an implication of $\is^b$.
By definition of $\is^b$, we have $a \neq b$ and $b \in \cl(a)$.
According to the construction of $\cl'$, we deduce $b \in \cl'(a)$ so $\is^b \subseteq \is'^b$.
Let $a \in \U$ such that $B \subseteq \cl(a)$ for some $B \in \Bm$.
By definition, $\cl(a) \subseteq \cl'(a) \nsubseteq B$ for every $B \in \Bp$.
Hence, by construction of $\cs'$, we deduce that $d \in \cl'(a)$ and $a \imp d \in \is'^b$ which proves the $\supseteq$ part of the equality.
We turn to the $\subseteq$ part.
Let $a \imp b$ be an implication of $\is'^b$, i.e., an implication derived from $b \in \cl'(a)$ and $a \neq b$.
We have two possible cases, $b = d$ or $b \neq d$.
Assume first that $b = d$.
By definition of $\cl'$, and since $\Bm$, $\Bp$ are dual in $(\cs, \subseteq)$, $d \in \cl'(a)$ implies that $B \subseteq \cl(a)$ for some $B \in \Bm$.
Thus, $a \imp b \in \{a \imp d \st B \subseteq \cl(a) \text{ for some } B \in \Bm\}$.
Remark that this holds in particular for the elements $a$ such that $\cl(a) \in \Bm$.
Now, suppose that $b \neq d$.
We show that $a \imp b \in \is^b$.
We have $b \in \cl'(a)$ by assumption.
By definition of $\cl'$, we also have $\cl(a) = \cl'(a) \setminus \{d\}$ and $b \in \cl(a)$ holds as $b \neq d$.
Hence, $a \imp b \in \is^b$.
We deduce that $\is'^b = \is^b \cup \{a \imp d \st B \subseteq \cl(a) \text{ for some } B \in \Bm \}$ is true as required.

It remains to prove that the non-binary part of $\is'_D$ is $\{K'_B \imp d \st B \in \Bm, \card{K'_B} \geq 2\}$.
We first show that $d$ is the unique element of $\U'$ that can have $D$-generators.
Let $c \in \U$ (so that $c \neq d$) and let $A$ be a non-trivial minimal generator of $c$ in $(\U', \cs')$.
Again using the construction of $\cl'$, we have $\cl^b(A) = \cl(A) = \cl'(A) \setminus \{d\}$ so that $A$ is a non-trivial minimal generator of $c$ in $(\U, \cs)$.
Since $(\U, \cs)$ is distributive, we deduce that $A$ must be a singleton.
In particular, $c$ cannot have $D$-generators in $(\U', \cs')$ unless $c = d$.

We proceed to the analysis of $D$-generators of $d$.
We prove that for every $B \in \Bm$, the unique minimal spanning set $K'_B$ of $\cl'(B)$ obtained from Proposition \ref{prop:Fp-CG} is a $D$-generator of $d$ provided $\card{K'_B} \geq 2$.
Note that if $\card{K'_B} < 2$, then $K'_B$ is a singleton since $(\U', \cs')$ is standard, and hence there is an implication of the form $K'_B \imp d$ in $\is'^b \subseteq \is'_D$.
Let us assume $\card{K'_B} \geq 2$.
First, we show that $K'_B$ is a minimal generator of $d$.
We have $\cl'(K'_B) = B \cup \{d\}$ where $B \in \Bm$.
Hence, for every $K \subset K'_B$, $d \notin \cl'(K) = \cl(K)$ by construction of $\cl'$.
Thus, $K'_B$ is a minimal generator of $d$.
We prove that it is now a $D$-generator, i.e., that for every $K \subseteq \U$ such that $\cl'^b(K) \subset \cl'^b(K'_B)$, $d \notin \cl'(K)$.
Hence, let $K \subseteq \U$ be such that $\cl'^b(K) \subset \cl'^b(K'_B)$.
We have $\cl'^b(K) = \cl(K) \subset \cl'^b(K'_B) = \cl(K'_B) = B$.
In particular, $\cl'^b(K) \in \idl \Bp$ in $(\cs, \subseteq)$.
Hence, $\cl'(K) = \cl(K)$ so that $d \notin \cl'(K)$.
We deduce that $K'_B$ is a $D$-generator of $d$ in $(\U', \cs')$ as expected.
Consequently, $\{K'_B \imp d \st B \in \Bm, \card{K'_B} \geq 2\} \subseteq \is'_D$ holds.

It remains to show that each $D$-generator of $A$ is the minimal spanning set $K'_B$ of some $B \in \Bm$, with $\card{K'_B} \geq 2$.
The fact that $\card{A} \geq 2$ follows from the definition of $D$-generators.
Since $d \in \cl'(A)$ and $d \notin A$, we deduce by construction of $\cl'$ that there exists $B \in \Bp$ such that $B \subseteq \cl(A) \subseteq \cl'^b(A) \subseteq \cl'(A)$.
In particular, $K'_B \subseteq \cl'^b(A)$ so that $\cl'^b(K'_B) \subseteq \cl'^b(A)$.
Since both $A$ and $K'_B$ are $D$-generators of $d$ by assumption and previous discussion, we deduce that $A = K'_B$ must hold.
As a consequence, $\{K'_B \imp d \st B \in \Bm, \card{K'_B} \geq 2\}$ precisely captures the $D$-generators of $d$ in $(\U', \cs')$.
This shows that the $D$-base is complete, and concludes the proof.
\end{proof}

\newpage

\section{Proof of Lemma \ref{lem:transition}} \label{app:transition}

In this section we give, for self-containment, a proof of Lemma \ref{lem:transition} within our phrasing.
The proof closely follows \cite{ennaoui2025polynomial, lucchesi1978candidate}.

\TRA*

\begin{proof}
We start with the if part.
Let $C = \cl_c^b((\cl_c^b(A) \setminus \cl_c^b(d)) \cup B)$ for some $A$ and $B \imp d$ as in the statement and assume that it does not include any set of $\cc{S}$.
Note that $c \notin C$ holds as $C \subseteq X_c$ by definition.
Since $d \in \cl_c(B)$ and $B \subseteq C$, we deduce $A \subseteq \cl_c(C)$ and hence $\cl_c(C) = \U_c$.
By Lemma \ref{lem:link-minimal-D} and since $C$ is $\cl_c^b$-closed, we deduce that $C$ contains a $D$-generator of $c$ that is not in $\cc{S}$.
This concludes this part of the proof.

We move to the only if part.
Assume $\cc{S} \neq \gen_D(c)$ with $\cc{S}$ non-empty.
Let $A' \in \gen_D(c) \setminus \cc{S}$.
We have $\cl^b(A') = \cl_c^b(A') \subseteq \U_c$.
Let $F_b$ be an inclusion-wise maximal $\cl_c^b$-closed set that includes $A'$ but not the $D$-generators of $\cc{S}$.
Since $\cc{S} \neq \emptyset$, $F_b \subset \U_c$.
On the other hand, $F_b$ contains a $D$-generator of $c$, so that $\cl_c(F_b) = \U_c$ by Lemma~\ref{lem:link-minimal-D}.
Moreover, $F_b$ is $\cl_c^b$-closed.
Therefore, there exists $B \imp d \in \is_c$ with $\card{B} \geq 2$ such that $B \subseteq F_b$ and $d \notin F_b$.
Consider $\cl_c^b(F_b \cup d)$.
By assumption on the maximality of $F_b$, there must be some $A \in \cc{S}$ such that $\cl_c^b(A) \subseteq \cl_c^b(F_b \cup d)$.
Given that $(\U_c, \cl_c^b)$ is standard by assumption and defines a distributive closure lattice, we have $\cl_c^b(F_b \cup d) = \cl_c^b(F) \cup \cl_c^b(d) = F_b \cup \cl_c^b(d)$.
Now, we have $\cl_c^b(A) \nsubseteq F_b$ but $\cl_c^b(A) \subseteq F_b \cup \cl_c^b(d)$.
Given that $B \subseteq F_b$, we thus deduce $\cl_c^b(A) \setminus \cl_c^b(d) \cup B \subseteq F_b$ and $\cl_c^b((\cl_c^b(A) \setminus \cl_c^b(d) \cup B) \subseteq \cl_c^b(F_b) = F_b$. 
But then, $A \subseteq \cl_c(\cl_c^b((\cl_c^b(A) \setminus \cl_c^b(d) \cup B))$ since $d \in \cl_c(B)$ as $B \imp d \in \is_c$.
We deduce that $\cl_c(\cl_c^b((\cl_c^b(A) \setminus \cl_c^b(d) \cup B)) = \U_c$ again by Lemma \ref{lem:link-minimal-D}.
As $\cl_c^b((\cl_c^b(A) \setminus \cl_c^b(d) \cup B)$ is $\cl_c^b$-closed it must contain a $D$-minimal key of $(\U_c, \cl_c)$, that is, a $D$-generator of $c$.
Because $\cl_c^b((\cl_c^b(A) \setminus \cl_c^b(d) \cup B) \subseteq F_b$, this $D$-generator does not belong to $\cc{S}$ as required.
This concludes the proof.
\end{proof}

\end{document}